\documentclass[preprint,10pt]{elsarticle}
\usepackage{amsmath}
\usepackage{amsthm,amsmath,amsfonts,amssymb,amscd,mathrsfs}
\usepackage{epstopdf}
\usepackage{txfonts}%,pxfonts}
\usepackage{supertabular,soul}
\usepackage[usenames,dvipsnames]{xcolor}
\usepackage{tikz, graphicx,color,geometry}
 \usepackage{multirow}
 \usetikzlibrary{arrows}
\usepackage[pdftex,
            pdfauthor={Dallas Smith},
            pdftitle={Automorphisms},
            pdfsubject={Graph Automorphisms},
            pdfkeywords={}]{hyperref}

\usepackage{bbm} %\mathbb numbers and other symbols
\usepackage{hyperref}
\usepackage{yfonts}
\usepackage{eucal}
\usepackage{overpic}
\usetikzlibrary{calc}

\usepackage{enumitem}

\newcommand{\comment}[1]{}

\def\big{\bigskip}

\newtheorem{theorem}{Theorem}
\newtheorem{thm}{Theorem}[section]

\newtheorem{lem}[thm]{Lemma}

\theoremstyle{definition}
\newtheorem{remark}[thm]{Remark}

\newtheorem{defn}[thm]{Definition}

\newtheorem{example}[thm]{Example}

\theoremstyle{remark}

\providecommand*{\propertyautorefname}{Property}

\let\oldmarginpar\marginpar
\renewcommand\marginpar[1]{\oldmarginpar[\raggedleft\footnotesize #1]%
{\raggedright\footnotesize #1}}

\graphicspath{ {C:\Users\Dallas\Documents\Webb\Tex} }

\begin{document}
\begin{frontmatter}

\date{\today}

\title{Hidden Symmetries in Real and Theoretical Networks}
\author[ds]{Dallas Smith}
\address[ds]{Department of Mathematics, Brigham Young University, Provo, UT 84602, USA, dallas.smith@mathematics.byu.edu}

\author[ben]{ and Benjamin Webb}
\address[ben]{Department of Mathematics, Brigham Young University, Provo, UT 84602, USA, bwebb@mathematics.byu.edu}

%%%%%%%%%%%%%%%%%%%%%%%%%%%%%%%%%%%%                                ABSTRACT                        %%%%%%%%%%%%%%%%%%%%%%%%%%%%%%%%%%%%%%%%%%%%%%

\begin{abstract}
Symmetries are ubiquitous in real networks and often characterize network features and functions.  Here we present a generalization of network symmetry called \emph{latent symmetry}, which is an extension of the standard notion of symmetry.  They are defined in terms of standard symmetries in a reduced version of the network.  One unique aspect of latent symmetries is that each one is associated with a \emph{size}, which provides a way of discussing symmetries at multiple scales in a network.  We are able to demonstrate a number of examples of networks (graphs) which contain latent symmetry, including a number of real networks. In numerical experiments, we show that latent symmetries are found more frequently in graphs built using preferential attachment, a standard model of network growth, when compared to non-network like (Erd{\H o}s-R\'enyi) graphs. Finally we prove that if vertices in a network are latently symmetric, then they must have the same eigenvector centrality, similar to vertices which are symmetric in the standard sense.  This suggests that the latent symmetries present in real-networks may serve the same structural and functional purpose standard symmetries do in these networks. We conclude from these facts and observations that \emph{latent symmetries} are present in real networks and provide useful information about the network potentially beyond standard symmetries as they can appear at multiple scales.

\end{abstract}

\begin{keyword}
Graph Automorphism \sep Network Symmetry \sep Isospectral Network Reduction\\
\end{keyword}

%\setcounter{tocdepth}{1}
%\tableofcontents
\end{frontmatter}

%%%%%%%%%%%%%%%%%%%%%%%%%%%%%%%%%%%%                                INTRODUCTION                        %%%%%%%%%%%%%%%%%%%%%%%%%%%%%%%%%%%%%%%%%%%%%%

\section{Introduction}

Real networks have a number of properties that distinguish them from many other graphs of interest. For instance, they tend to have right-skewed degree distributions, high clustering coefficients, and the ``small-world" property, etc. \cite{newman2010networks}. Additionally, real networks generally contain a significant number of symmetries \cite{dettmann2017symmetric}.  Consequently, many real networks have a large symmetry group \cite{Mac2008}.  It is important to study these symmetry groups for a number of reasons. First, understanding network symmetry helps us better understand the formation of particular networks \cite{symI}. Symmetries can also provide information about vertex function. For instance, it has been observed that two symmetric nodes can play the same role in a network, which is thought to increase network robustness \cite{macarthur2009spectral}. In the case of networks dynamics and function, symmetries are known to be important to the processes of synchronization or partial synchronization \cite{antoneli2006symmetry}.

Beyond these standard symmetries, %In addition exact symmetric subgraphs,
``near'' symmetries also naturally occur in real networks \cite{symI}. Even though these approximate symmetries are not represented in the symmetry group of the network, they still have an effect on network behavior, both in form and in function  \cite{symI}. For this reason, there has been a number of attempts to weaken the notion of structural symmetry.  By \emph{structural symmetry}, we mean there exists a permutation of the network's vertices that does not change the network structure. A near symmetry can be described in terms of properties invariant under some other network transformation.  As an example the network's vertex degree could be maintained as the network's topology is transformed \cite{holme2006detecting}.

Notions of stochastic symmetry have also been established \cite{symI} to characterize near symmetries in real networks.  In this framework one chooses a statistical ensemble of networks which are similar but not exactly identical and assigns a probability measure to them.   This allows one to quantify approximate symmetry and associate characteristics of similar networks. This weakening of the notion of symmetry has led to the study of symmetry groupoids, which can be been used to create synchronization in dynamical networks \cite{stewart2004networking}, \cite{olver2015symmetry}, \cite{stewart2003symmetry}.

In this paper we propose a very different extension of the notion of symmetry called \emph{latent symmetry}. Latent symmetries are derived from structural symmetry in a particular reduced version of the network. There are many ways to reduce a network, such as removing edges in a specified way or collapsing chosen subnetworks into single vertices to ``course grain'' the network (for instance, see \cite{xiao2008network}).  When finding latent symmetries, we do not use any of these techniques but use what is referred to as an \emph{isospectral graph reduction} \cite{thebook} to reduce the size of the network. That is, a latent symmetry is a standard structural symmetry in an isospectral reduction of the network. This specific method is chosen because it preserves the spectral properties of a network, i.e. the eigenvalues and eigenvectors associated with the network. The motivation for using a isospectral reduction is that a network's spectral properties encode various structural characteristics, including graph connectivity, vertex centrality, and importantly symmetry (see \cite{barrett2017equitable}, \cite{brouwer2011spectra}, \cite{hahn2013graph}, \cite{chungspectral}, among others).  Additionally, for dynamical networks, stability and other dynamic properties depend on the spectrum of the network \cite{thebook}, \cite{almendral2007dynamical}.

An important property of latent symmetries that is, to the best of our knowledge, not possessed by any other type of symmetry is that it has a sense of \emph{scale}.  That is, we can define a \emph{measure of latency} for any latent symmetry, which one can think of as how deep the symmetry is buried within the network. This is of particular interest since many real networks are known to have a \emph{hierarchical structure} in which statistically significant substructures known as \emph{motifs} are repeated at multiple scales throughout the network \cite{alessandro2007large}. Our findings suggest that not only are motifs to be found at multiple levels in a real network, but also symmetries.  Thus one can study a network's hierarchical structure of symmetries to better understand the interplay of network structure and function, in particular a network's multiple levels (scales) of redundancy.

This paper is organized as follows.  In section 2, after describing standard network symmetries, we precisely define a \emph{latent symmetry} and give a number of examples of such symmetries.  We then define the \emph{measure of latency} of a latent symmetry. In section 3 we give examples of real-world networks which exhibit latent symmetries. In section 4 we prove various spectral properties regarding latent symmetries, including showing that if two nodes are latently symmetric then they have the same eigenvector centrality i.e. the same importance in the network using this metric. Finally, in section 5 we further argue that latent symmetries have relevance in the real world by showing that they are more likely to occur in networks generated using preferential attachment, rather than networks that are randomly generated, e.g. Erd{\H o}s-R\'enyi graphs.  This is significant because preferential attachment models capture many characteristics of real networks, while Erd{\H o}s-R\'enyi graphs do not \cite{price1976general}.

%^%%%%%%%%%%%%%%%%%%%%%%%%%                           Section 2                         %%%%%%%%%%%%%%%%%%%%%%%%%%%%%%%%%%%%%%%%%

\section{Network Symmetries}\label{sec:2}
 The standard method used to describe the topology of a network is a graph. Here, a \emph{graph} $G=(V,E,\omega)$ is composed of a \emph{vertex set} $V$, an \emph{edge set} $E$, and a function $\omega$ used to weight the edges $E$ of the graph. The vertex set $V$ represents the \emph{elements} of the network, while the edges $E$ represent the links or \emph{interactions} between these network elements. In some networks, it is useful to define a direction to each interaction.  This is the case in which an interaction between two network elements influences one but not the other. For instance, in a citation network, in which network elements are papers and edges represent whether one paper cites another, papers can only cite papers that have already been written. Thus each edge has a clearly defined direction.  We model this type of network as a \emph{directed} graph in which each edge is directed from one network element to another. If this does not apply, the edges are not directed and we have an \emph{undirected} graph.

 In this paper we are most interested in \emph{strongly connected} graphs, meaning for any two vertices in the graph there exists a path in the graph which start at one vertex and ends at the other.  For graphs that are not strongly connected, we often only consider the largest strongly connected component of the graph. Our reasons for doing this are essentially pragmatic since many real-world networks are extremely large consisting of a million or more vertices.  However, our theory applies to all networks whether strongly connected or not.

 The weights of the edges given by $\omega$ measure the \emph{strength} of these interactions. Some examples of weighted networks include: social networks where weights corresponds to the frequency of interaction between actors, food web networks where weights measure energy flow, or traffic networks where weights measure how often roads are used \cite{newman2010networks}. Here we consider networks with positive real-valued edge weights because they represent the majority of weighted networks considered in practice.  Though it is worth mentioning that the theory we present throughout the paper is valid for more general edge weights, e.g. complex-valued or more complicated weights (see for instance \cite{thebook}).

Let $G=(V,E,\omega)$ be a weighted graph on $n$ vertices representing a network.  Its \emph{weighted adjacency matrix}, $M=M(G)$, is an $n\times n$ matrix whose entries are given by $$M_{ij} = \left\{
        \begin{array}{ll}
            \omega(e_{i,j})\neq 0 & \text{if  }e_{i,j}\in E \\
            0 & \text{otherwise}
        \end{array}
    \right.
$$
where $e_{i,j}$ is the edge from vertex $i$ to vertex $j$. An \emph{unweighted} graph can be considered to be a special case of a weighted graph where all edge weights are equal to $1$.  Moreover, the weighted adjacency matrix of an undirected graph is symmetric since each edge can be thought of as a directed edge oriented in both directions.  Figure \ref{symmetry5} gives an example of an unweighted, undirected graph with its corresponding adjacency matrix. We note here that an adjacency matrix which corresponds to a strongly connected graph is by definition an \emph{irreducible matrix}.

It is also worth noting that there is a one-to-one relation between weighted graphs (networks) and their corresponding weighted adjacency matrices $M\in\mathbb{R}^{n\times n}$ meaning that there is no more information presented in one than the other.  Often it is more convenient to work with matrices instead of graphs, though both are useful ways to represent network structure. Graphs are typically used for network visualization while matrices are better suited for network analysis \cite{newman2010networks}. Throughout the paper we will use graphs and matrices without ambiguity to refer to the ``graph of the network'' and the ``matrix associated with the network.''

 The particular type of structure we consider in this paper is the notion of a graph symmetry, which can be understood via graph automorphisms.  These symmetries have received considerable attention in the literature  \cite{Mac2008}, \cite{symI}, \cite{xiao2008emergence}.  Intuitively, a graph automorphism describes how parts of a graph can be interchanged in a way that preserves the graph's overall structure.  In this sense these \emph{parts}, i.e., subgraphs, are symmetrical and together constitute a graph symmetry.  For example, consider the graph in Figure \ref{symmetry5}.  Here, it is easy to visually identify the symmetry between the yellow vertices 6 and 8, since transposing them would not change the graph's structure. Formally, a \emph{graph automorphism} of $G$ is defined to be a permutation $\phi: V \to V$ of the graph's vertices $V$ that preserves weights between the network's vertices.
 \begin{figure}
 \begin{center}
 \begin{tabular}{ll}
\raisebox{-.5\height}{
\begin{overpic}[scale=.4]{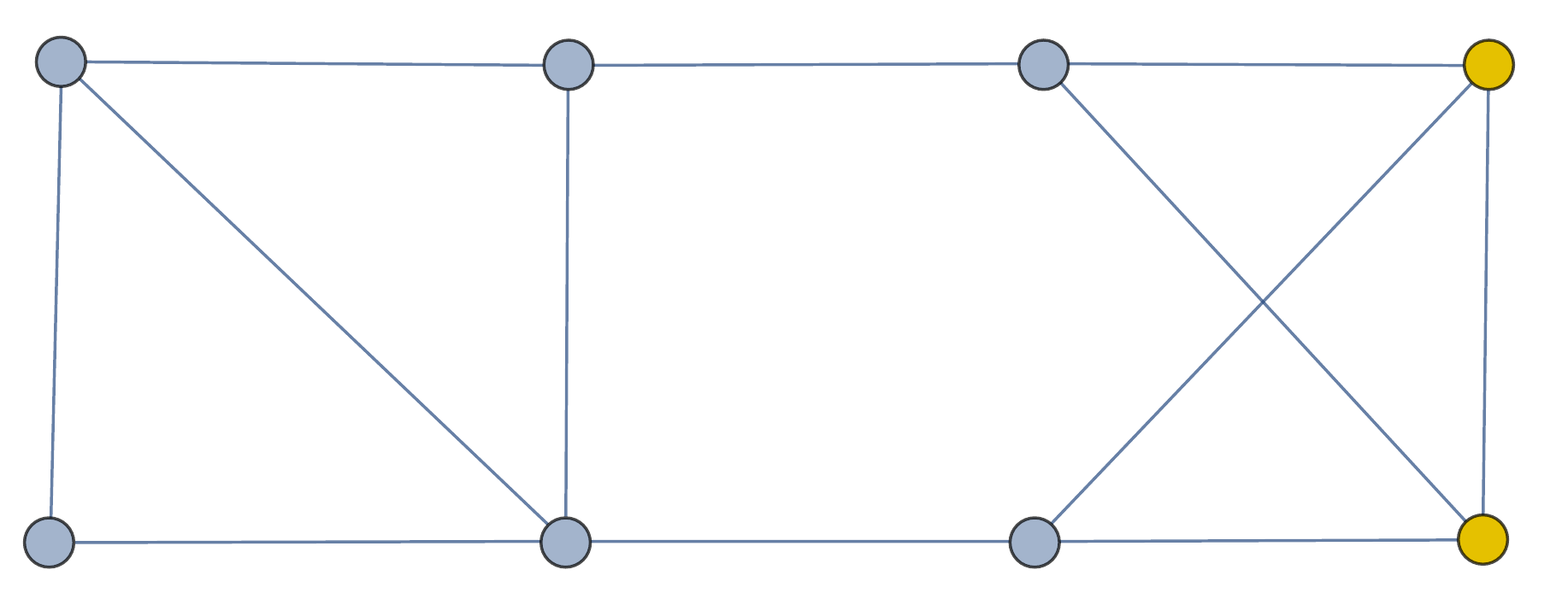}
	\put(50,-3){$G$}
	\put(38,6){$1$}
	\put(71,5){$2$}
	\put(38,35){$3$}
	\put(5,35.5){$4$}
	\put(4.5,6){$5$}
	\put(97.5,35){$6$}
	\put(70,35){$7$}
	\put(97.5,6){$8$}
  \end{overpic}}
  & \ \ \   $M(G)=\left( \begin{matrix}
0&1&1&1&1&0&0&0\\
1&0&0&0&0&1&0&1\\
1&0&0&1&0&0&1&0\\
1&0&1&0&1&0&0&0\\
1&0&0&1&0&0&0&0\\
0&1&0&0&0&0&1&1\\
0&0&1&0&0&1&0&1\\
0&1&0&0&0&1&1&0
\end{matrix}\right)
$ \
\end{tabular}
\end{center}
\caption{An example of an unweighted, undirected graph $G$ and its corresponding adjacency matrix $M(G)$. The graph has the symmetry given by the automorphism $\phi=(68)$.  The symmetric vertices 6 and 8 are highlighted yellow.}\label{symmetry5}
\end{figure}
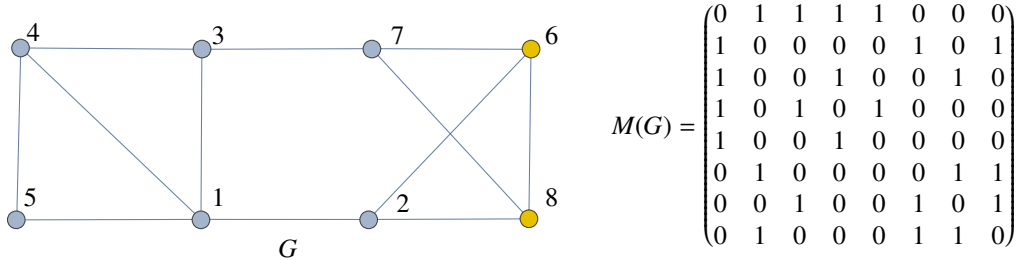
    \begin{defn}\label{def:sym} \textbf{(Graph Automorphism)}
    An \emph{automorphism} $\phi$ of a weighted graph $G=(V,E,\omega)$ is a permutation of the graph's vertex set $V$ such that the weighted adjacency matrix $M$ satisfies $M_{ij} = M_{\phi(i) \phi(j)}$ for each pair of vertices $i$ and $j$ in $E$.
    \end{defn}
In the case of an unweighted graph, this definition is equivalent to saying vertices $i$ and $j$ are adjacent in $G$ if and only if $\phi(i)$ and $\phi(j)$ are adjacent in $G$. A collection $S$ of vertices in $V$ are \emph{symmetric} if for any two elements $a,b$ in $S$ there exists an automorphism $\phi$ of $G$ such that $\phi (a)=b$. As an example, the vertices 6 and 8 in Figure \ref{symmetry5} are symmetric since the permutation $\phi$ that transposes 6 and 8 and fixes all other vertices of $G$ (written in permutation cycle notation as $\phi=(68)$), is an automorphism of $G$.

Structural symmetries in networks are important as they can provide information about network robustness as well as the function of specific vertices \cite{macarthur2009spectral}. Often some of the same type of information can be extracted from a set of vertices which are ``nearly'' symmetric. There are a number of ways which have been proposed to precisely define a ``near'' symmetry \cite{symI}. Our method involves finding structural symmetries in a reduced version of the network.  The network is reduced in a way that preserves spectral properties of the network's adjacency matrix.  We are interested in preserving the spectrum of the network since there are a number of important network characteristics which can be determined from its spectrum.

To make this idea more precise, we need the notion of an \emph{isospectral graph reduction}, which is the method we will use to reduce the underlying graph structure of a network. This is a graph operation which produces smaller graph with essentially the same set of eigenvalues as the original graph. This method for reducing the graph associated with a network can be formulated both for the graph and equivalently for the adjacency matrix associated with the network, i.e. an \emph{isospectral graph reduction} and an \emph{isospectral matrix reduction}, respectively. Both types of reductions will be useful to us.

For the sake of simplicity we begin by defining an isospectral matrix reduction. For this reduction we need to consider matrices of rational functions. The reason is that, by the Fundamental Theorem of Algebra, a matrix $A\in\mathbb{R}^{n\times n}$ has exactly $n$ eigenvalues including multiplicities. In order to reduce the size of a matrix while at the same time preserving its eigenvalues we need something that carries more information than just scalars. The objects we will use to preserve this information are rational functions. The specific reasons for using rational functions can be found in \cite{thebook}, Chapter 1.

For a matrix $M\in\mathbb{R}^{n\times n}$ let $N=\{1,\ldots,n\}$. If the sets $R,C\subseteq N$ are proper subsets of $N$, we denote by $M_{RC}$ the $|R| \times |C|$ \emph{submatrix} of $M$ with rows indexed by $R$ and columns indexed by $C$. We denote the subset of $N$ not contained in $S$ by $\bar{S}$, that is $\bar{S}$ is the \emph{complement} of $S$. We let $\mathbb{W}^{n\times n}$ be the set of $n\times n$ matrices whose entries are rational functions  $p(\lambda)/q(\lambda)\in\mathbb{W}$, where $p(\lambda)$ and $q(\lambda)\neq0$ are polynomials with real coefficients in the variable $\lambda$ with no common factors. The isospectral reduction of a square real-valued matrix is defined as follows.

\begin{defn}\label{def:isored} \textbf{(Isospectral Matrix Reduction)}
The \emph{isospectral reduction} of a matrix $M\in\mathbb{R}^{n\times n}$ over the proper subset $S\subseteq N$ is the matrix
\[
\mathcal{R}_S(M) = M_{SS} - M_{S\bar{S}}(M_{\bar{S}\bar{S}}-\lambda I)^{-1} M_{\bar{S}S}\in\mathbb{W}^{|S|\times|S|}.
\]
\end{defn}

 The eigenvalues of the matrix $M=M(\lambda)\in\mathbb{W}$ are defined to be solutions of the \emph{characteristic equation}
\[
\det(M(\lambda)-\lambda I)=0,
\]
which is an extension of the standard definition of the eigenvalues for a matrix with complex entries. By way of notation we let $\sigma(M)$ denote the set of eigenvalues of $M$ including multiplicities. An important aspect of an isospectral reduction is that the eigenvalues of the matrix $M$ and the eigenvalues of its isospectral reduction $\mathcal{R}_S(M)$ are essentially the same, as described by the following theorem \cite{thebook}.

\begin{theorem}\label{thm:maintheorem}\textbf{(Spectrum of Isospectral Reductions)} For $M\in\mathbb{R}^{n\times n}$ and a proper subset $S\subseteq N$, the eigenvalues of the isospectral reduction $\mathcal{R}_S(M)$ are
$$\sigma\big(\mathcal{R}_S(M)\big)=\sigma(M)-\sigma(M_{\bar{S}\bar{S}}).$$ \vspace{-1.25cm}
\end{theorem}

That is, when a matrix $M$ is isospectrally reduced over a set $S$, the set of eigenvalues of the resulting matrix is the same as the set of eigenvalues of the original matrix $M$ after removing any elements which are eigenvalues of the submatrix $M_{\bar{S}\bar{S}}$.

Phrased in terms of graphs, if the graph $G=(V,E,\omega)$ with adjacency matrix $M$ is isospectrally reduced over some proper subset of its vertices $S\subseteq V$ then the result is the  reduced graph $\mathcal{R}_S(G)=(S,\mathcal{E},\mu)$ with adjacency matrix $\mathcal{R}_S(M)$. Hence,
\[
\sigma\big(\mathcal{R}_S(G)\big)=\sigma(G)-\sigma(G|\bar{S}).
\] \vspace{-1.25cm}

\noindent where eigenvalues of a graph are the eigenvalues of a graph's adjacency matrix and where $G|\bar{S}$ denotes the subgraph of $G$ restricted to the vertices not contained in $S$. It is worth noting that the matrix $M$ and the submatrix $M_{\bar{S}\bar{S}}$ often have no eigenvalues in common, in which case the spectrum is unchanged by the reduction, i.e. $\sigma(\mathcal{R}_S(M))=\sigma(M)$.

Using isospectral reductions we can define a generalization of the notion of a graph symmetry.

\begin{defn}\label{def:ls}\textbf{(Latent Symmetries)}
We say a graph $G$ has a \emph{latent symmetry} if there exists a subset of vertices which are symmetric in some isospectral reduction $\mathcal{R}_S(G)$ of $G$.
\end{defn}

\begin{figure}
\begin{center}
  \begin{overpic}[scale=.4]{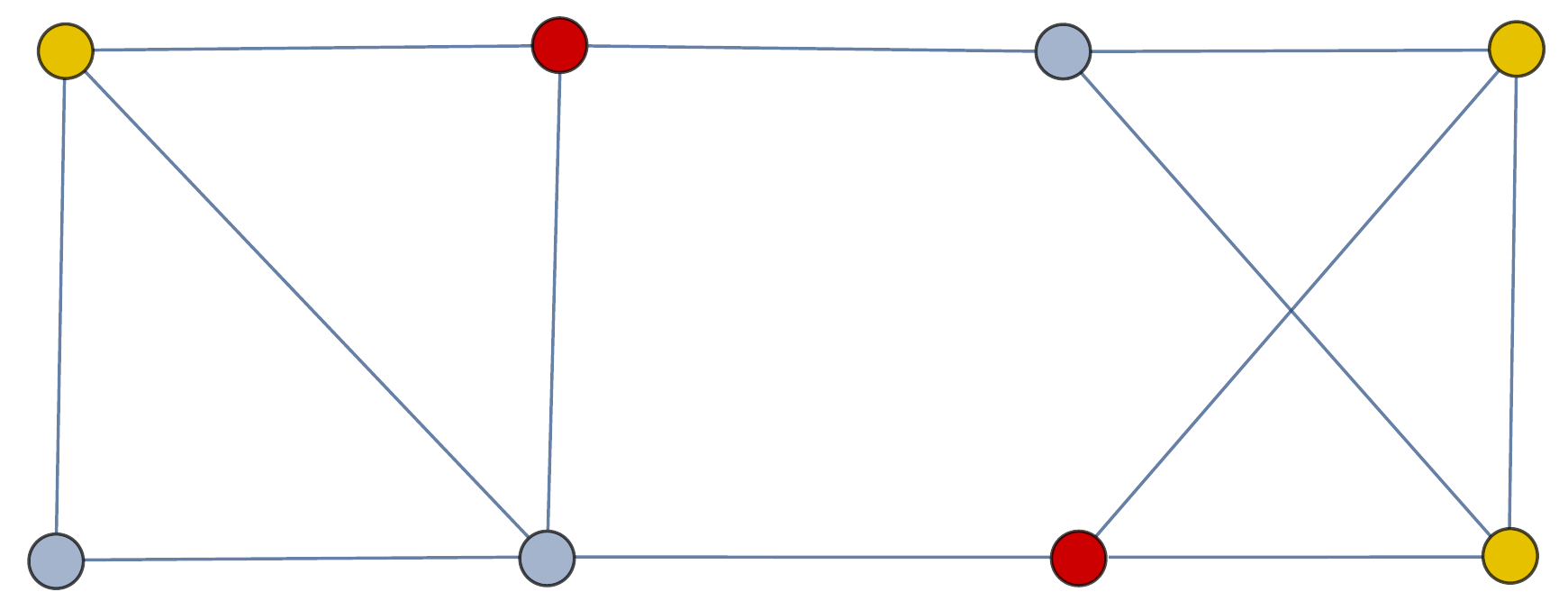}
	\put(50,-3){$G$}
	\put(37.5,5){$1$}
	\put(72.5,3.9){$2$}
	\put(38.5,36.5){$3$}
	\put(6,36.5){$4$}
	\put(4.5,5){$5$}
	\put(98.5,37.5){$6$}
	\put(70,36.5){$7$}
	\put(98,6){$8$}
  \end{overpic}\qquad
  \raisebox{.2\height}{
  \begin{overpic}[scale=.35]{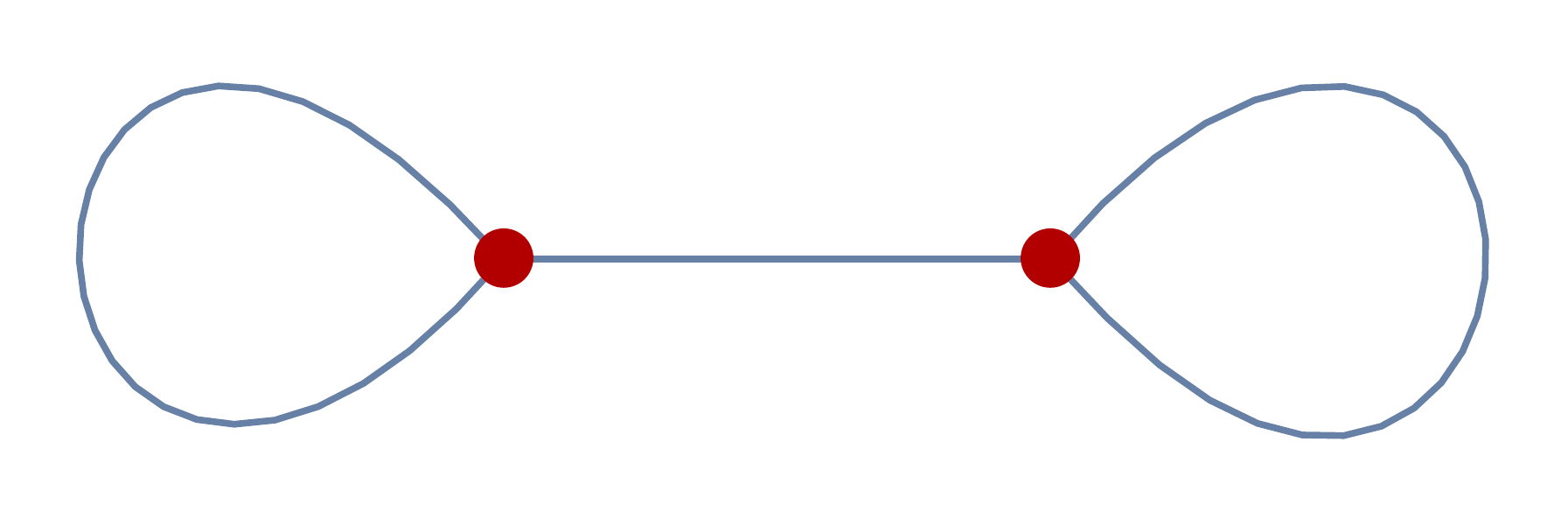}
	\put(40,2){$\mathcal{R}_{\{2,3\}}(G)$}
	\put(31,20){$3$}
	\put(65,19){$2$}
	\put(7,31){$\frac{3\lambda}{\lambda^2-\lambda-2}$}
	\put(75,31){$\frac{3\lambda}{\lambda^2-\lambda-2}$}
	\put(42,20){$\frac{\lambda+2}{\lambda^2-\lambda-2}$}
  \end{overpic} }
\end{center}
\caption{(Left) The undirected graph $G$ from Figure \ref{symmetry5} which has both standard and latent symmetries. Red vertices 2 and 3 are latently symmetric.  Yellow vertices 6 and 8 have a standard symmetry between them, but 4 is latently symmetric to both of 6 and 8. (Right) The isospectral reduction of the top graph over vertices 2 and 3, showing the latent symmetry between these two vertices }\label{fig:0}
\end{figure}
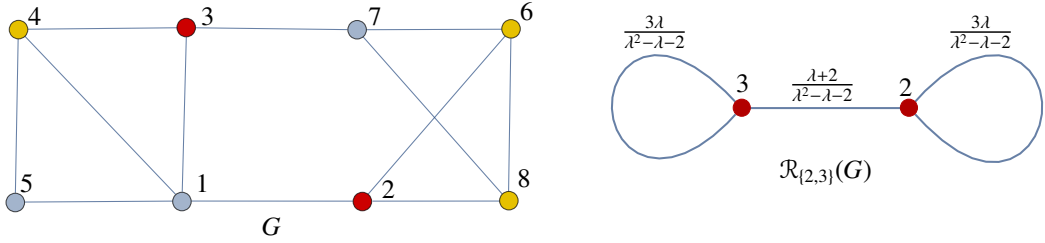
 The reason we refer to such symmetries as latent symmetries is that they are difficult to see before the graph reduction is performed, thus they are in some sense \emph{hidden} within the network. For the remainder of the paper, structural symmetries as defined in Definition \ref{def:sym} will be referred to as \emph{standard symmetries} to distinguish them from the latent symmetries defined in Definition \ref{def:ls}. We note here that standard symmetries are a subset of latent symmetries since reducing a graph $G=(V,E,\omega)$ over its entire set of vertices preserves the graph, i.e. $\mathcal{R}_V(G)=G$.

\begin{example} The graph in Figure \ref{fig:0} is an example of a graph with both standard and latent symmetries.  In this figure, colors correspond to groups of vertices which are latently symmetric e.g. vertices 2 and 3.  We note that the yellow vertices 6 and 8 form a standard graph symmetry since transposing the two vertices, (i.e. switching their  labels), does not change the graph structure of the graph $G$. When $G$ is reduced over vertices 4 and 6 (or 4 and 8), the resulting reduced graph contains a standard symmetry, i.e. 4 and 6 (4 and 8) are latently symmetric.  Also reducing $G$ over the red vertices 2 and 3 results in a graph with symmetry as shown on the right of Figure \ref{fig:0}.
\end{example}

\begin{figure}
\begin{center}
\includegraphics[scale=.35]{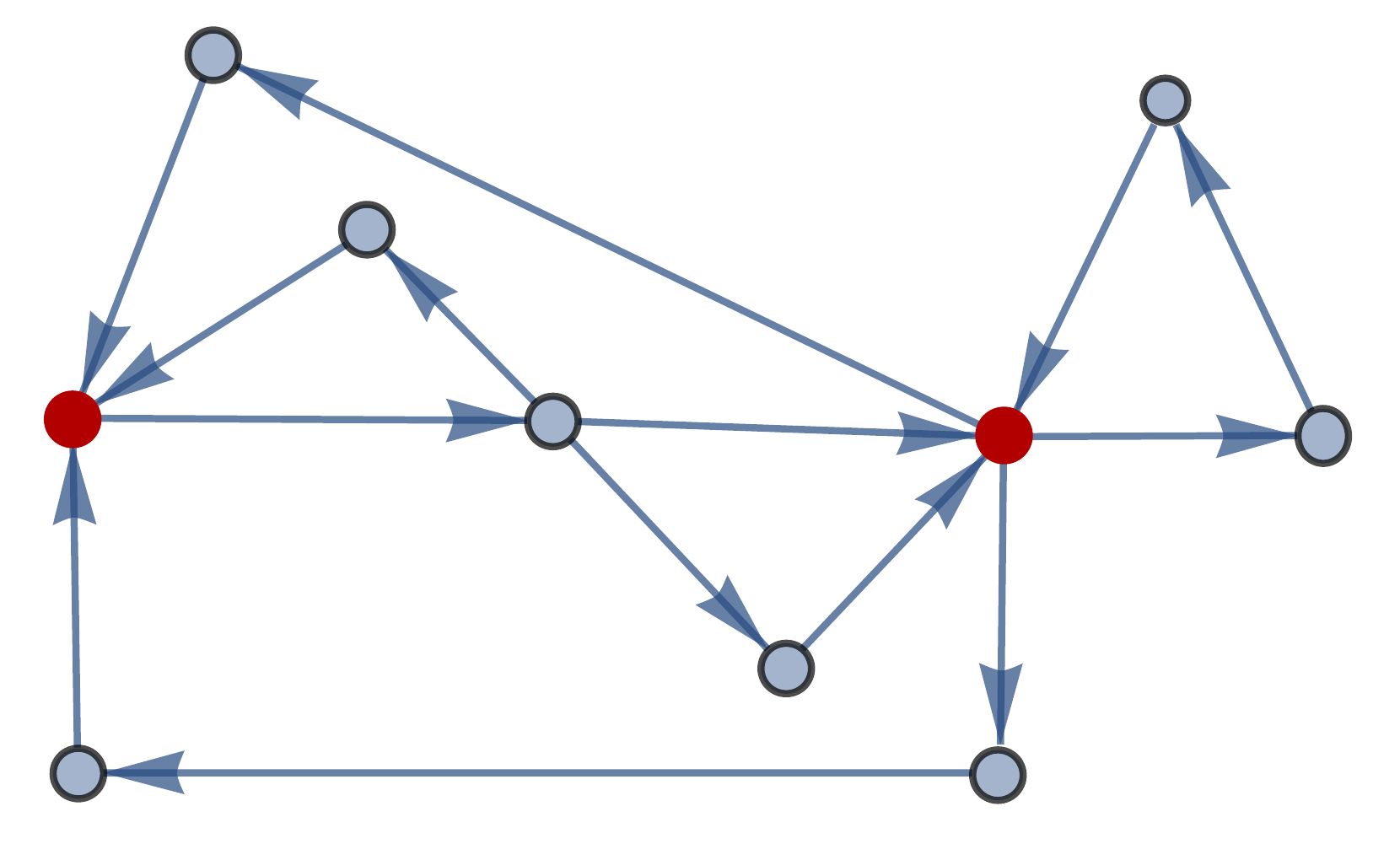}	\qquad
  \begin{overpic}[scale=.35]{0rednl.pdf}
    \put(-70,-5){\large $G$}
	\put(40,0){\large $\mathcal{R}_{\{1,4\}}(G)$}
	\put(32,19){\small $1$}
	\put(63.5,19){\small $4$}
   \put(-105.5,26){\small $1$}
	\put(-65,30){\small $2$}
       \put(-77,40){\small $3$}
	\put(-34.5,23){\small $4$}
       \put(-48.7,11){\small $5$}
	\put(-88,54){\small $6$}
           \put(-35,4){\small $7$}
	\put(-105,4){\small $8$}
       \put(-14,26){\small $9$}
	\put(-25,49){\small $10$}
	\put(-1,17){$\frac{1}{\lambda^2}$}
	\put(96,17){$\frac{1}{\lambda^2}$}
	\put(45,20){$\frac{\lambda+1}{\lambda^2}$}
  \end{overpic}
\caption{In the graph $G$ (left) vertices 1 and 4 are latently symmetric as can be seen by reducing $G$ to $\mathcal{R}_{\{1,4\}}(G)$ (right).}\label{fig:directedex}
\end{center}
\end{figure}
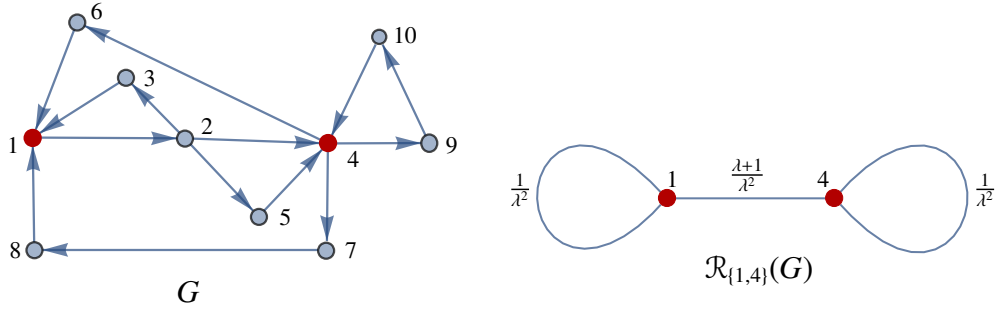

Some properties of standard symmetries extend to latent symmetries (see, for example, the results in section \ref{sec:4}).  Like standard symmetries, latent symmetries are \emph{transitive}. By this we mean that if there exists a latent symmetry between vertices $a$ and $b$ and a latent symmetry between the vertices $b$ and $c$ in a graph, there must be a latent symmetry between vertices $a$ and $c$.  We note, however, in this scenario there is no guarantee that there exists of subset of vertices $S$ such that $a,b$, and $c$ are all is symmetric in $\mathcal{R}_S(G)$, i.e. $a$, $b$ and $c$ may not be latently symmetric \emph{as a set}.  This is in contrast to standard symmetries where for any set of vertices that are pairwise symmetric, there must exist an automorphism for which all these vertices lie in the same orbit i.e. they are all symmetric as a set. In the setting of standard symmetris this is proved by noting the composition of two automorphisms is an automorphism.

Before moving on in it is also worth mentioning that the definition for latent symmetries works for both directed and undirected graphs. The following is an example of a directed graph with a latent symmetry, which we include here to give the reader some intuition for how latent symmetries can arise. For this we note that a \emph{path} in a directed graph $G=(V,E,\omega)$ is a sequence of distinct vertices, $\{v_0,v_1,\dots,v_n\}\subseteq V$ such that the set of directed edges $\{(v_i,v_{i+1}) \ \ | \ \ 0\leq i\leq n-1\}$ is contained in $E$ and a \emph{cycle} is path for which $v_0=v_n$.
 \begin{example}\label{ex:dir}
Consider the directed graph $G$ on the left of Figure \ref{fig:directedex}. Because all cycles in this graph contain one of the red vertices (labeled 1 and 4), we can write a finite list of paths and cycles that both begin and end on these vertices.  We observe that the paths from 1 to 4 are \{1,2,4\} and \{1,2,5,4\}, while all paths from 4 to 1 are \{4,6,1\} and \{4,7,8,1\}. Thus, there are the same number of paths with the same lengths from 1 to 4 as there are from 4 to 1. Also there is only one cycle from vertex 1 to itself (not including 4); namely \{1,2,3,1\}, and only one cycle from vertex 4 to itself (not including 1); namely \{4,9,10,4\}. Both of these cycles have length 4. This symmetry in number and length of paths and cycles guarantees that a symmetry will appear after reducing the graph over these two vertices. That is, vertices 1 and 4 are latently symmetric.  This can be seen in the reduced graph $\mathcal{R}_{\{1,4\}}(G)$ in which there is an automorphism between these two vertices (transposition).

When a graph contains cycles which do not contain any vertices in the reducing set $S$, it is not possible to write down all the paths and cycles as we did for the graph $G$.  In this case it becomes much more difficult to construct and identify latent symmetries. This is the case for even small undirected graphs. To highlight the unintuitive nature of latent symmetry in this case as well as give a sense of the variety of latent symmetries that occur in undirected graphs, Figure \ref{fig:lots} depicts six more graphs with examples of latent symmetries.  It is worth mentioning that by an exhaustive search we have found the smallest undirected graphs which contain a latent symmetry which is also not a standard symmetry has eights vertices (eg. Figure \ref{fig:lots}, top left).

\end{example}

\begin{figure}
\begin{center}
\includegraphics[scale=.5]{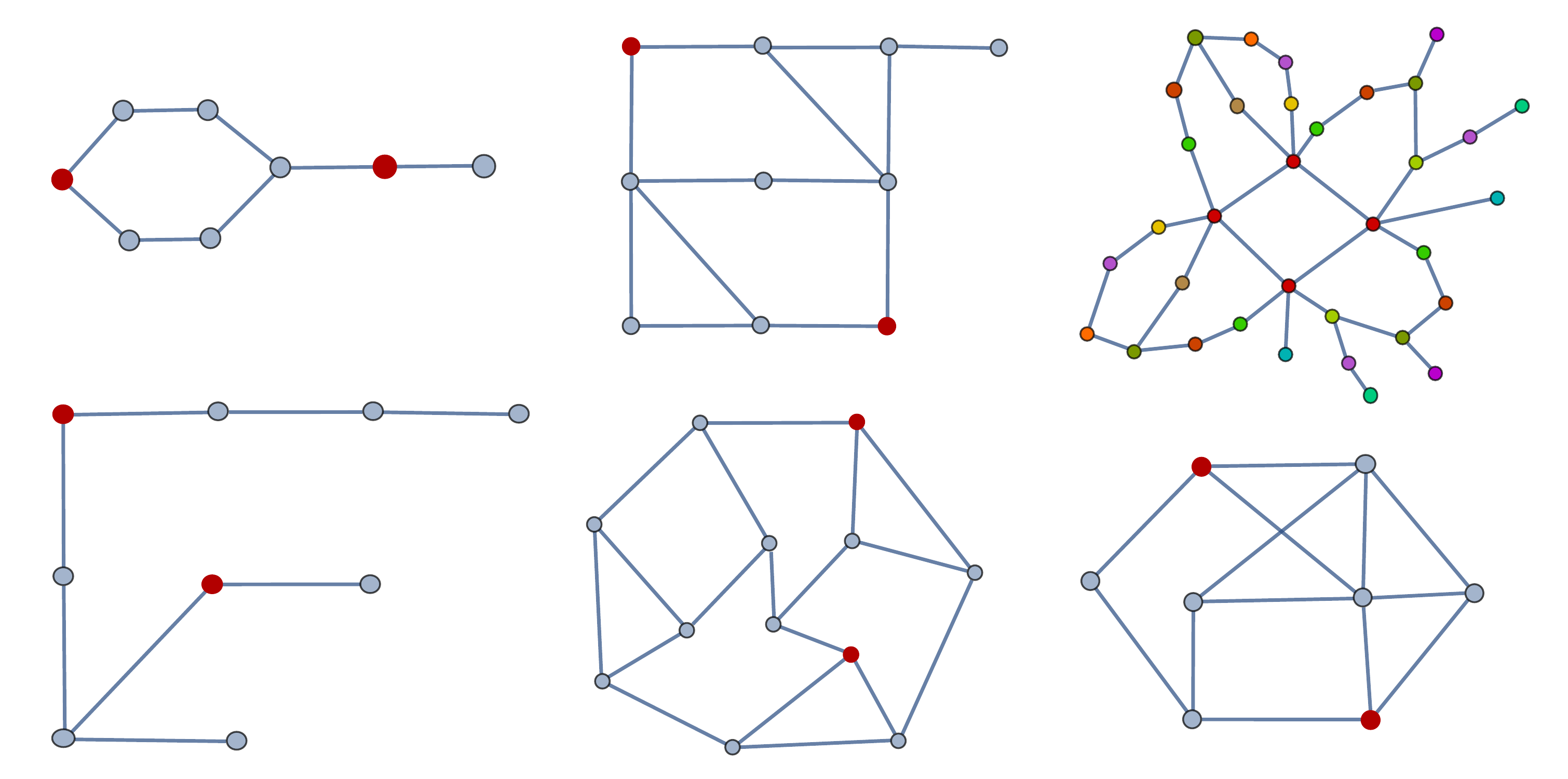}
\caption{Six examples of undirected graphs which contain latent symmetries. Vertices of the same color in a graph correspond to latent symmetries. (Top right) A smallest example of a graph with a latent symmetry but no standard symmetries (fewest edges and vertices), (Top middle) a graph which is symmetric without the pendant vertex on top, which turns the standard symmetry into a latent symmetry, (Top right) a graph where every vertex is latently symmetric to at least one other, though there are no standard symmetries, (Bottom left) a tree, (Bottom middle) a 3-regular graph, (Bottom right) a non-planar graph.}\label{fig:lots}
\end{center}
\end{figure}

Another useful concept we can explore regarding latent symmetries is the scale at which the symmetry is found within the network. %degree to which a latent symmetry is a standard symmetry.

\begin{defn}\label{MOL}\textbf{(Measure of Latency)} \
Let $G=(V,E,\omega)$ be a graph with $n$ vertices and let $S$ be a subset of its vertices which are latently symmetric. This latent symmetry can be said to have a \emph{measure of latency} $\mathcal{M}$, defined as $$\mathcal{M}(S)=\frac{n-|T|}{n-|S|}$$ where $T\subseteq V$ is a maximal set of vertices such that the vertice $S$ are symmetric in $\mathcal{R}_T(G)$.
\end{defn}
From this definition it is clear that $0\leq \mathcal{M}(S)\leq 1$ since $|T|\geq |S|$.  Moreover, if the vertices $S$ are symmetric in the unreduced graph, i.e. are symmetric in the standard sense, then $T=S$ and $\mathcal{M}(S)=0$.  On the other hand, if there is no possible choice of a vertex set $T$ for which $\mathcal{R}_T(G)$ has a symmetry between the vertices in $S$, except for $S=T$, then $\mathcal{M}(S)=1$.  This is the most hidden a latent symmetry can be since it requires reducing the entire graph to the set $S$ before the symmetry can be seen.  We note that although this is an interesting measure, it can be computationally difficult to find since it requires finding the largest possible reducing set under which a symmetry forms.

\begin{figure}
\begin{overpic}[scale=.25]{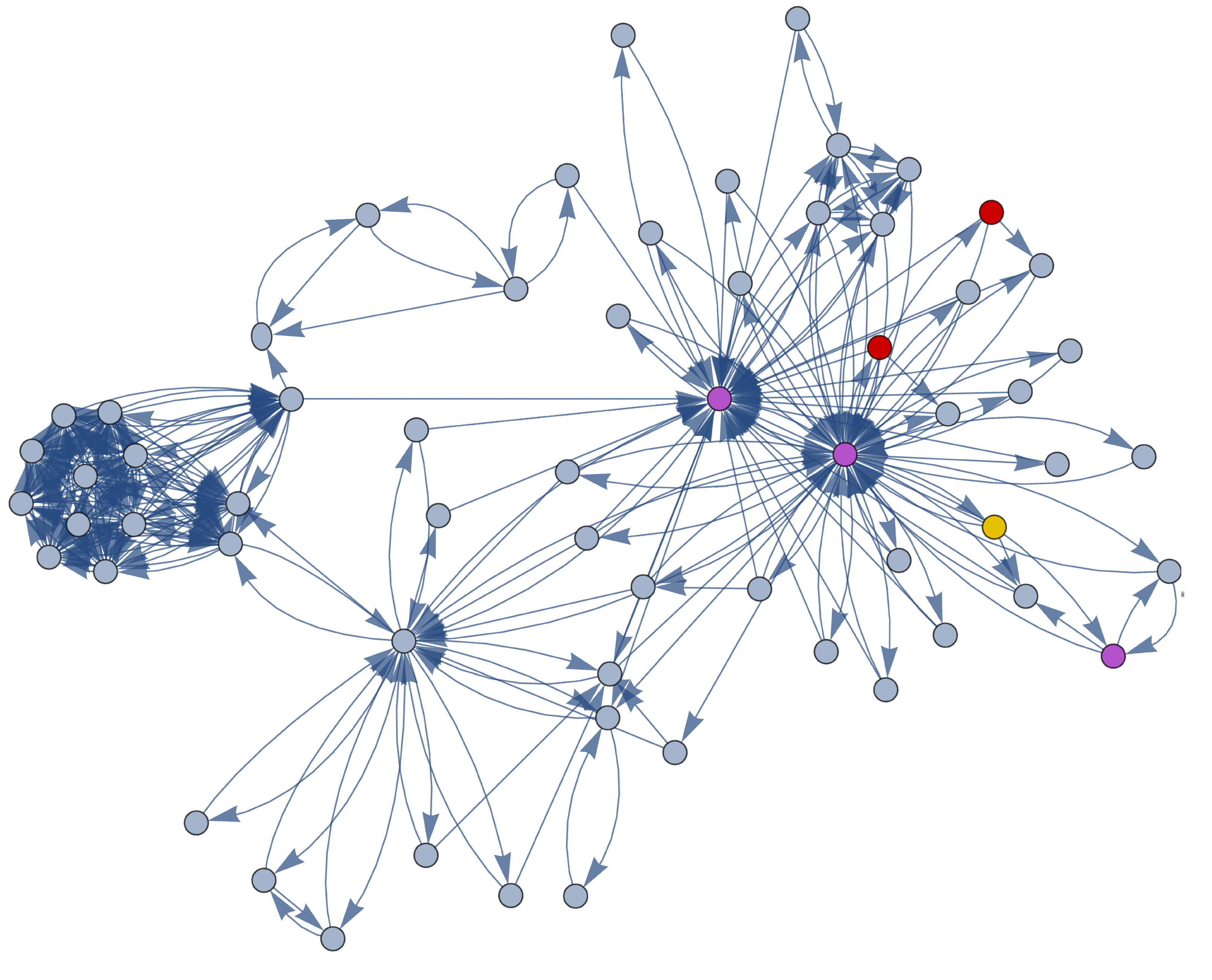}
\put(83.5,63){\tiny 1}
\put(94,24.5){\tiny 7}
\put(83.7,36.5){\tiny 3}
\put(88,58){\tiny 4}
\put(78.9,47.7){\tiny 5}
\put(83.5,28.4){\tiny 6}
\put(73.5,53){\tiny 2}
\end{overpic} \ \ \
\begin{overpic}[scale=.3]{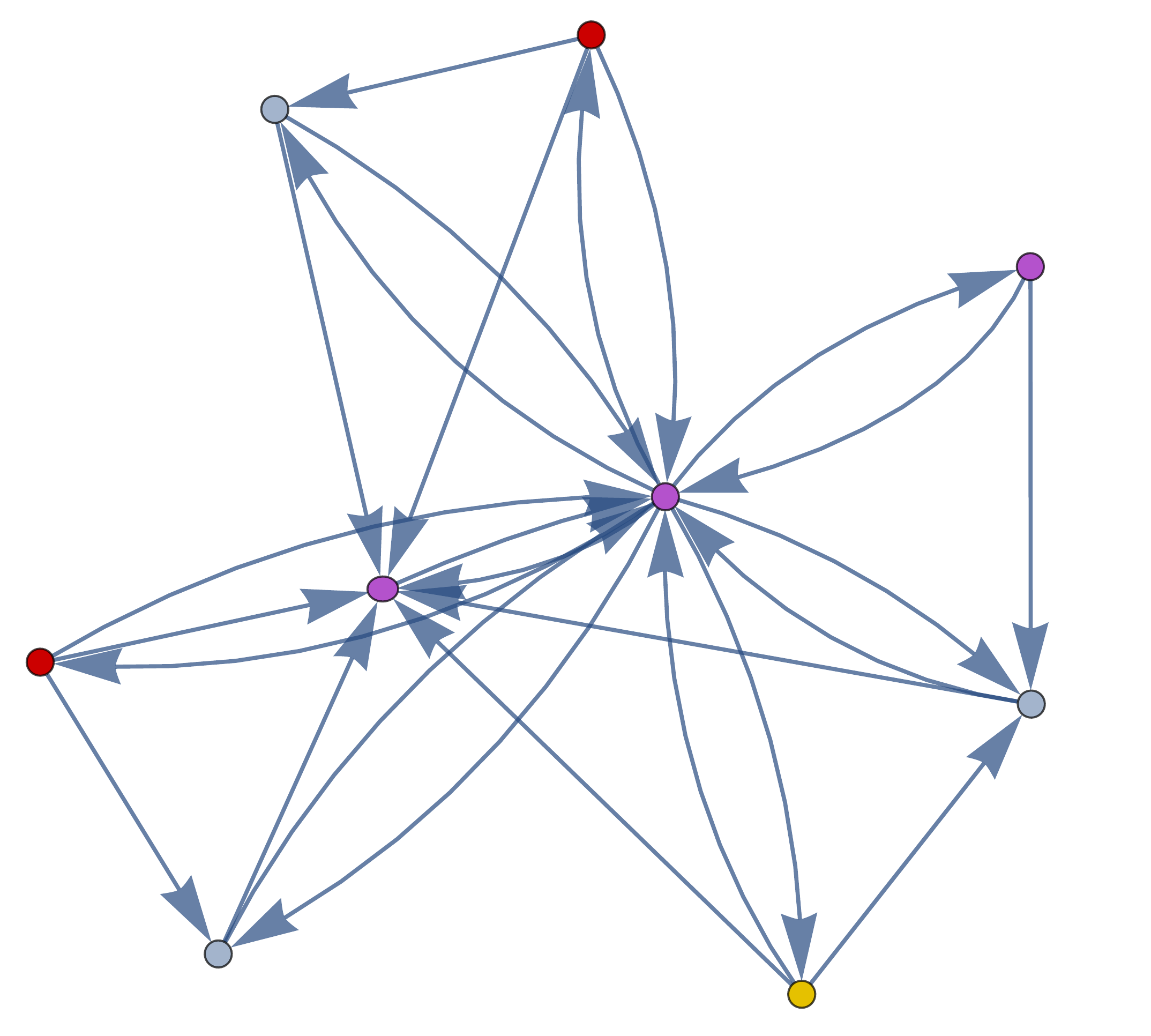}
\put(54,86){\small 1}
\put(93.2,65){\small 7}
\put(74,0){\small 3}
\put(25,84){\small 4}
\put(19,1){\small 5}
\put(93.2,26){\small 6}
\put(-1.5,32){\small 2}
\end{overpic}
\caption{(Left) A network representation of the largest strongly connected component of all Wikipedia webpages in the ``logic puzzle'' category \cite{wiki}.  Vertices represent webpages while the direct edges represent hyperlinks between them. (Right) A subgraph of this network. Red vertices are symmetric and the yellow vertex is latently symmetric with the two red vertices. Purple vertices are nodes which have edges that are not displayed.  }\label{fig:att}
\end{figure}
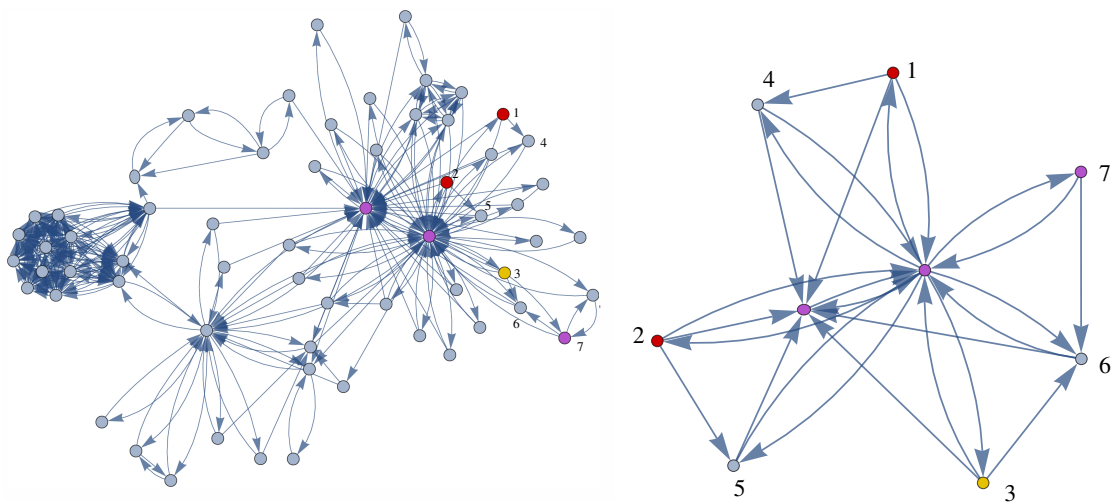

\begin{example}
For the graph $G$ in Figure \ref{fig:0} we have shown that vertices 2 and 3 are latently symmetric as they are symmetric in the reduced graph $\mathcal{R}_{\{2,3\}}(G)$.  However, we actually do not need to reduce to such a small graph to see this symmetry. In fact, 2 and 3 are symmetric in the graph $\mathcal{R}_{\{2,3,4,8\}}(G)$, but are not symmetric in any reduction over five or more vertices. Thus the largest reducing set $T$ in which this symmetry appears must contain four elements.  Thus, $\mathcal{M}(\{2,3\})=\frac{n-|T|}{n-|S|}=\frac{8-4}{8-2}=2/3$.
\end{example}

The measure of latency we give to a network symmetry gives the symmetry a size or a scale within the network. This is reminiscent of one of the hallmarks of real networks in which specific structures, known as \emph{motifs}, are found at multiple scales within the network \cite{newman2006modularity}, \cite{leskovec2008statistical}. In the following section we investigate whether latent symmetries occur in real-world networks and specifically whether these occur at different scales.

%%%%%%%%%%%%%%%%%%%%%%%%%%%%%%%%%%%%                                SECTION  3                      %%%%%%%%%%%%%%%%%%%%%%%%%%%%%%%%%%%%%%%%%%%%%%%%%%%%%%%%

\section{Latent Symmetries in Real Networks}\label{sec:3}

The first question one might have concerning latent symmetries is the extent to which they are actually observed in real network data.   In this section we present two very different real-world networks which contain latently symmetric vertices.
\begin{example}\label{ex:Logic}
Consider the web graph shown in Figure \ref{fig:att} (left) which represents all Wikipedia pages contained in the category \emph{``Logic Puzzles''} in August 2017. Each vertex represents a webpage and directed edges represent hyperlinks between webpages \cite{wiki}. The two red vertices are symmetric in this graph, while the yellow vertex is pairwise latently symmetric with the two red vertices. Figure \ref{fig:att} (right) shows a subgraph of the left graph to more clearly demonstrate the path structure causing the latent symmetry. In this second graph purple vertices do not display all of their connections. All vertices labeled 1-7 represent puzzles published by \emph{``Nikoli.''}  We can see that puzzles 1, 2 and 3 all have an edge pointing to another puzzle (4, 5 and 6 respectively).  However, 6 also has vertex 7 pointing to it. This breaks the symmetry between 3 and 1 (also between 3 and 2).  The latent symmetry is still present between 3 and 1 (also between 3 and 2) since the same paths are available for traversing the graph from 1 to 3 or from 3 to 1. Thus the latent symmetry is highlighting a common feature in these three puzzle that a standard symmetry search would overlook. For the vertices which are latently symmetric one can calculate their measure of latency to be $\mathcal{M}(\{1,3\})=1/30\approx 0.033.$

We mention that in the next section will show that by using the metric of eigenvector centrality, vertices 1 and 3 have the same importance to this network.  In fact, any set of vertices that are latently symmetric have the same eigenvector centrality (see Theorem \ref{thm:EC}).
\\

\end{example}

\begin{example}
A second example of a latent symmetry in real-world network data is in the metabolic network for the cellular processes in \emph{Arabidopsis thaliana}, a eukaryotic organism \cite{UND}. This is a biological network of chemical reactions, which is a very different type of network than the one considered in Example \ref{ex:Logic}.  Figure \ref{Euk} (left) shows the the largest strongly connected component of this network. Here vertices represent cellular substrates (as well as intermediary states) and edges represent metabolic pathways.  The red vertices highlighted in the figure (left) have a latent symmetry between them.  The right graph is a subgraph of the left where vertex color is preserved.  Yellow vertices represent substrates for which all of their edges from the original network are displayed, while purple vertices do not have all of their edges displayed. We can see that many of these vertices are almost symmetric, meaning many of the vertices appear to have a corresponding vertex with a similar local path structure.  It is worth noting that between the red vertices there exists the same number of paths of the same length. As in the previous example, this suggests some kind of structural similarity exists between these vertices that a search for standard symmetry would overlook.  These vertices are again very close to being symmetric, which is quantized by their measure of latency $\mathcal{M}(S) \approx .0231.$

\begin{figure}
\includegraphics[scale=.25]{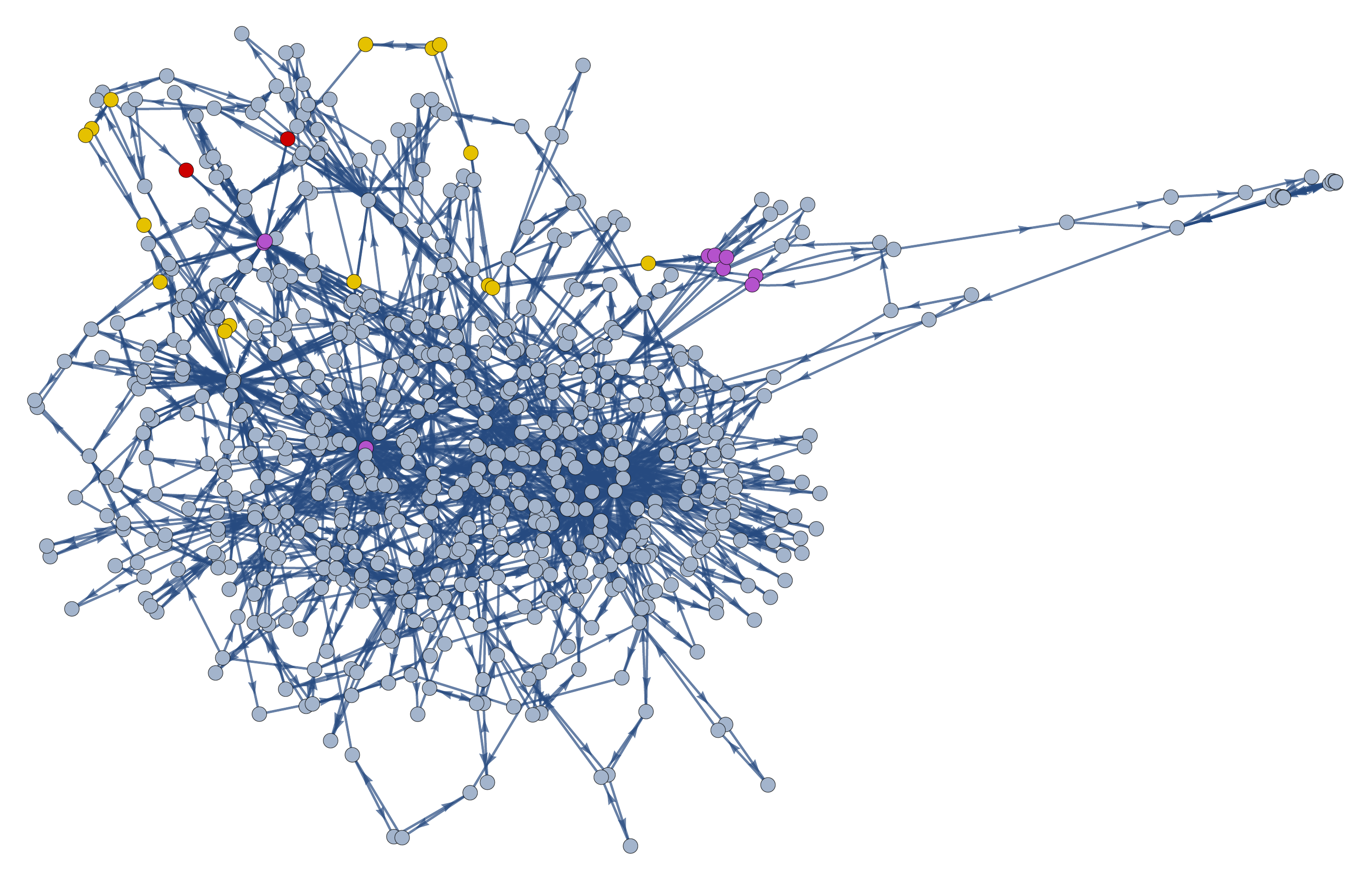}
\includegraphics[scale=.17]{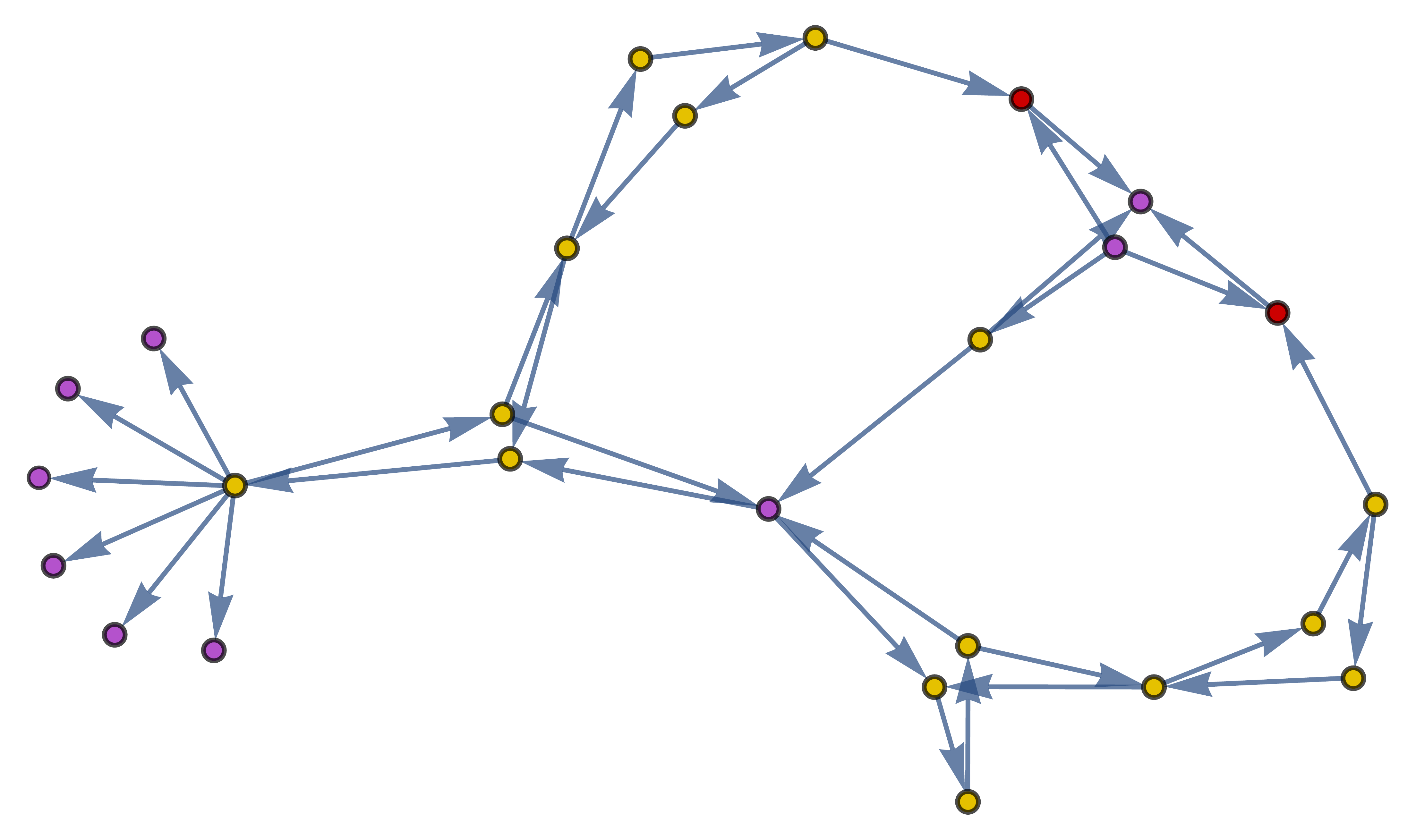}
\caption{Left: Metabolic network of the eukaryotic organism Arabidopsis Thaliana \cite{UND}.  Right: Subnetwork of left network. In both latently symmetric vertices are colored red. Yellow vertices have been drawn with all their original connections intact whereas purple vertices are missing edges from the original graph.}\label{Euk}
\end{figure}
\end{example}

Currently it is an open question whether latently symmetric vertices have similar or complementary functions within a network. The answer is likely that both are possible and is presumably network dependent. An important point is that once a latent symmetry has been found, an expert in the field which studies the specific network may be able to better answer these questions.

It is also worth emphasizing that both of the real networks we consider in this section have symmetries at different scales.  That is, both have standard and latent symmetries which are not standard. We can think of this as symmetries at multiple levels which leads to what one could refer to as a \emph{hierarchy of symmetries}. It is currently an open question as to how such symmetries might be distributed at various scales through a typical real network.

Going through real-network data, we find it more difficult to find examples of undirected networks than directed networks with latent symmetries.  In fact, these symmetries appear to become more rare as the number of vertices in the undirected network gets large. We explore this further using numerical simulations in section \ref{sec:5}.

\section{Eigenvector Centrality}\label{sec:4}

In the previous section we presented examples of latent symmetries in real-world networks. We now present evidence to support our claim that latent symmetries capture some type of hidden structure in a network.  Specifically we prove that when two vertices are latently symmetric they must have the same \emph{eigenvector centrality}, which is a standard measure of how important a vertex is compared to the other vertices in the network. This result suggests that the notion of a latent symmetry is indeed a natural extension of the standard notion of symmetry since vertex symmetries in the standard sense have the same eigenvector centrality and is therefore an important structural concept that can be used to analyze real networks.

Eigenvector centrality is a widely used metric in network analysis \cite{newman2010networks}. In fact, it is the basic principle used by ``Google'' to rank the webpages in the World Wide Web. It is calculated  by ranking the vertices by the value of the corresponding entry in the leading eigenvector of the network's adjacency matrix, where the leading eigenvector is the eigenvector associated with the matrix' largest eigenvalue. Not all graphs have a leading eigenvector. However, essentially all real world networks satisfy the conditions of the Perron-Frobenius theorem which guarantees the existence of a leading eigenvector for the network \cite{newman2010networks}.

To define eigenvector centrality, suppose a network $G$ is represented by a matrix $M$, and $\lambda_0$ is its largest eigenvalue with eigenvector $\mathbf{x}$. The \emph{eigenvector centrality} of a vertex $i\in G$ is the $i^{th}$ entry in $\mathbf{x}$, or $x_i$.

In order to extend this concept of eigenvector centrality to reduced networks, we first need to extend the notion of eigenvectors to the class of matrices with rational function entries. If $R(\lambda)=\mathcal{R}_S(M)$ is an $|S|\times|S|$ matrix with rational function entries and $\lambda_0\in \sigma(M)-\sigma(M_{\bar{S}\bar{S}})$ then $R(\lambda_0)$ is defined, that is each rational function in the matrix $R(\lambda)$ is defined at $\lambda_0$. Hence, $R(\lambda_0)$ is simply an $|S|\times |S|$ real-valued matrix \cite{thebook}.  We say the vector $\mathbf{v}\in \mathbb{C}^{n\times 1}$ is an eigenvector corresponding to $\lambda_0$ if $(R(\lambda_0)-\lambda_0 \mathbf{I})\mathbf{v}=0$.  We note that if the network we are considering is strongly connected with non-negative edge weights then $R(\lambda_0)$ is always defined and $\lambda_0=\rho(M)$, where $\rho(M)=\max \{|\lambda|:\lambda \in \sigma (M) \}$ is the \emph{spectral radius} of $M$ \cite{horn1990matrix}.

It is easy to show that vertices always have the same eigenvector centrality if there is a standard symmetry between them \cite{koschutzki2005centrality}. As previously mentioned, one interesting property of latent symmetries is that if two vertices in a network are latently symmetric, then they will also have the same eigenvector centrality.  That is, using the metric of eigenvector centrality, latently symmetric vertices have the same importance in the network.  This suggests that latent symmetries reveal something as important as the presence of a standard symmetry in the underlying structure of a network.

To prove this result we will first need the following theorem that relates the eigenvectors of $\mathcal{R}_S(M)$ to the eigenvectors of $M$, which is found in \cite{duarte2015eigenvectors}.

\begin{thm}\label{thm:reduction}\textbf{(Theorem 1 in \cite{duarte2015eigenvectors})} Suppose $M\in\mathbb{R}^{n\times n}$ and $S\subseteq N$. If $(\lambda,\mathbf{v})$ is an eigenpair of $M$ and $\lambda\notin \sigma(M_{\bar{S}\bar{S}})$ then $(\lambda,\mathbf{v}_S)$ is an eigenpair of $\mathcal{R}_S(M)$,  where $\mathbf{v}_S$ is the \emph{projection} of $\mathbf{v}$ onto $S$, i.e. $\mathbf{v}_S$ are the components of $\mathbf{v}$ indexed by $S$.
\end{thm}

This theorem states that under an isospectral reduction the eigenvectors of the matrix are preserved in the sense that the eigenvectors of the reduced matrix are simply the projection of the original eigenvector onto the vertices the network was reduced over.

In order to show that two latently symmetric vertices have the same eigenvector centrality, we reduce the graph over these vertices to create a graph that contains a symmetry.  The leading eigenvector of the reduced graph can then be related the eigenvectors of the original graph using the following result.

\begin{thm}\label{thm:EC}\textbf{(Eignvector Centrality and Latent Symmetries)}
Let $G=(V,E,\omega)$ be a graph (directed or undirected) with nonnegative edge weights which is strongly connected. If there exists a set $L\subseteq V$ of vertices that are latently symmetric, then these vertices all have the same eigenvector centrality.
\end{thm}
\begin{proof}
Because $G$ is strongly connected, the Perron-Frobenius Theorem, guarantees that $M=M(G)$ has a largest simple eigenvalue, $\lambda_0 \in \mathbb{R}$, which is equal to the spectral radius $\rho(M)$. Lemma \ref{v} in the Appendix  guarantees that $\lambda_0\notin \sigma(M_{\bar{S}\bar{S}})$ for any vertex set $S$.  The eigenvector $\mathbf{v}$ associated to $\lambda_0$ is by definition the leading eigenvector of $M$. Also, the isospectral reduction of a strongly connected graph  must also be strongly connected, thus $\mathcal{R}_S(G)|_{\lambda=\rho(M)}$ also has a leading eigenvector.

Now recall that a set $L\subseteq V$ of vertices of $G$ is latently symmetric if when $G$ is reduced over a set $S\supseteq L$, the vertices in $L$ are symmetric in the resulting reduced graph. Thus
%we can used Theorem \ref{thm:centrality} to guarantee that
all vertices in $L$ have the same eigenvector centrality in $\mathcal{R}_S(G)$ since symmetric vertices have the same eigenvector centrality, meaning each has the same value in the leading eigenvector, $\mathbf{v}_S$.  Next we use Theorem \ref{thm:reduction}. Since symmetric vertices correspond to entries in $\mathbf{v}_S$ with the same value, they must also correspond to equal entries in $\mathbf{v}$, the leading  eigenvector for the original matrix. This is because $\mathbf{v}_S$ is simply a projection of $\mathbf{v}$.  Thus vertices which are latently symmetric must have the same eigenvector centrality.

\end{proof}
\begin{figure}
\begin{center}
\begin{overpic}[scale=.4]{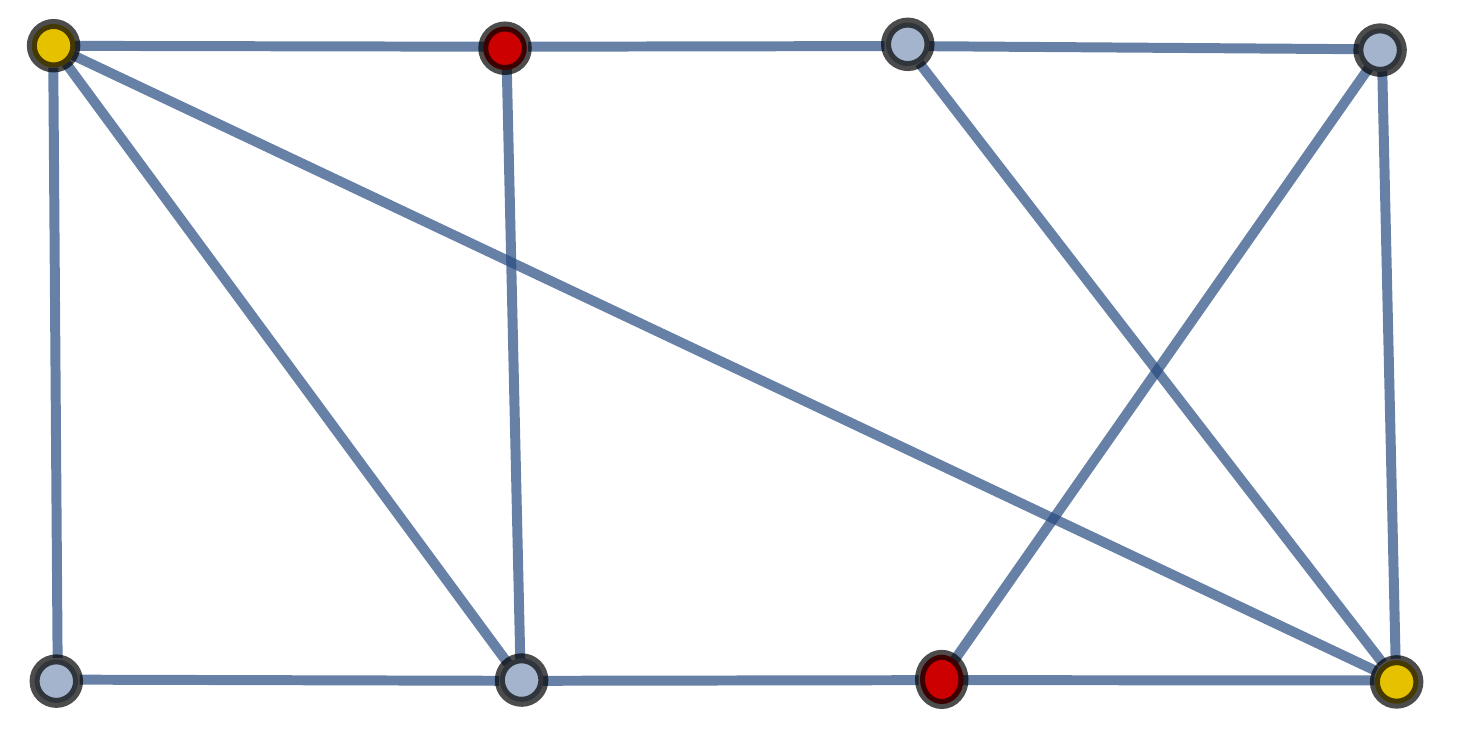}
\put(91,43){\Huge $\bullet$}
\put(59,44){\Huge $\bullet$}
\put(50,-5){$G$}
\put(97,49){$b$}
\put(65.7,50){$a$}
\end{overpic}\qquad
\raisebox{.2\height}{
\begin{overpic}[scale=.4]{0rednl.pdf}
\put(29,13.1){\Huge $\bullet$}
\put(64,13.1){\Huge $\bullet$}
\put(40,-5){$\mathcal{P}_{\{a,b\}}(G)$}
\put(34,19){$a$}
\put(62,18.5){$b$}
\put(8,29.5){$ 1.3399$}
\put(77,29.5){$ 1.3399$}
\put(44,11){$ 2.014$}
\end{overpic}}
\caption{(Left) The two large black vertices $a$ and $b$ in the graph $G$ are not latently symmetric, but have the same eigenvector centrality.  Red and yellow vertices are represent two latent symmetries in the graph. (Right) The Perron complement graph of $G$ over $\{a,b\}$, $\mathcal{P}_{\{a,b\}}(G)$.}\label{fig:new}
\end{center}
\end{figure}
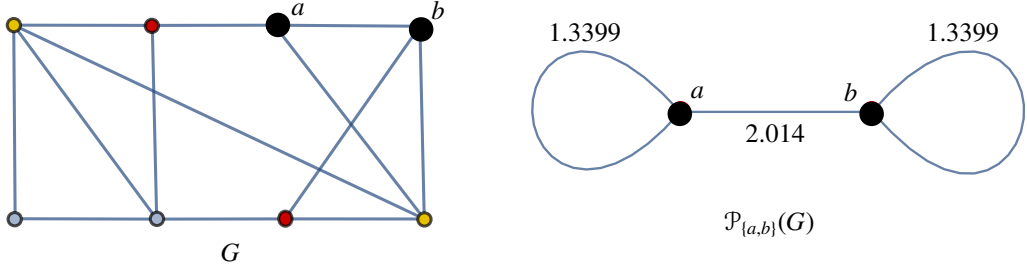
If two vertices are latently symmetric then by Theorem \ref{thm:EC}, both vertices have the same eigenvector centrality.  However, it is worth emphasizing that even though being latently symmetric is a sufficient condition for having the same eigenvector centrality, it is \emph{not} a necessary condition. For example, consider the graph in Figure \ref{fig:new}.  The two large black vertices have the same eigenvector centrality, although these vertices are \emph{not} latently symmetric.  We can, however, strengthen the conclusion of Theorem \ref{thm:EC} for the case of \emph{undirected} graphs.  We do this by searching for symmetry in the graph associated with the \emph{Perron complement} of the networks adjacency matrix (as opposed to the isospectral reduction of the network which was defined in section \ref{sec:2}).  The Perron complement of a matrix is defined as follows. \\

\begin{defn}{\textbf{(Perron complement)}} \cite{meyer1989uncoupling}
Let $M\in \mathbb{R}^{n\times n}$ be a nonnegative, irreducible matrix with \emph{spectral radius} $\rho (M)$. If $S\subseteq N$ then the matrix $$\mathcal{P}_S(M)=M_{SS}-M_{S\bar{S}}(M_{\bar{S}\bar{S}}-\rho(M)I)^{-1}M_{\bar{S}S}\in \mathbb{R}^{|S|\times |S|}$$
is the \emph{Perron complement} of M over S.
\end{defn}

Note that $\mathcal{P}_S(M)$ is the isospectral reduction of $M$ over $S$ in which we let $\lambda=\rho(M)$ so that ${\mathcal{P}}_S(M)$ is a real-valued matrix.
%For nonnegative irreducible matrix $M$ its \emph{leading eigenvector} is the unique nonnegative unit vector $\mathbf{v}$ such that $M\mathbf{v}=\rho(M)\mathbf{v}$.
Regarding the Perron complement, the following results hold.
\begin{thm}\label{thm:irr}  \textbf{(Theorem 2.2 and 3.1 in \cite{meyer1989uncoupling})}
Let $M\in \mathbb{R}^{n\times n}$ be a nonnegative irreducible matrix and let $S\subseteq N$. Then
\\
(i) the Perron complement $\mathcal{P}_S(M)$ is also a non-negative and irreducible matrix with the same spectral radius, i.e. $\rho (M)=\rho(\mathcal{P}_S(M))$; and \\
(ii) if $\mathbf{v}$ is the leading eigenvector of $M$ then its projection $\mathbf{v}_S$ is the leading eigenvector of $\mathcal{P}_S(M)$.
\end{thm}

Property (i) in Theorem \ref{thm:irr} shows that the spectral radius is unaffected by this reduction. Property (ii) is analogous (and a result of) Theorem \ref{thm:reduction}, which tell us that the eigenvectors in the Perron complement are projections of the eigenvectors of the original matrix.

We note here that analogous to isospectral reduction for graphs in section \ref{sec:2}, we can also define the \emph{Perron complement graph} for a graph $G$ with adjacency matrix $M$ and vertex subset $S$ denoted by $\mathcal{P}_S(G)$, which is the graph whose weighted adjacency matrix is $\mathcal{P}_S(M)$. The Perron complement graph can in some sense be more convenient for analysis since it must have real-valued edge weights, as opposed to the rational function weights which result from isospectral reductions.

However, our interest in the Perron complement in that using the properties of Theorem \ref{thm:irr}, we can give a necessary and sufficient condition under which two vertices of a graph have the same eigenvector centrality.

\begin{thm}\label{thm:centrality2}
\textbf{(Eigenvector Centrality in the Perron Complement Graph)} Let $G$ be an undirected connected graph.  Vertices $i,j$ have the same eigenvector centrality if and only if they are symmetric in $\mathcal{P}_{\{i,j\}}(G)$.
\end{thm}
Theorem \ref{thm:centrality2} is proved in the Appendix. As an example of this theorem the large black vertices $a$ and $b$ in Figure \ref{fig:new} (left) have the same eigenvector centrality. In the Perron complement graph, $\mathcal{P}_{\{a,b\}}(G)$, shown in Figure \ref{fig:new} (right), we can see a symmetry appears between these two vertices.

\begin{remark}
We note here that Theorem \ref{thm:centrality2} is stated for two vertices and does not extend to any other set of vertices. Specifically if two vertices of some network $G$ have the same eigenvector centrality then they are symmetric in the Perron complement graph of $G$ over these two vertices.  The conclusion does not hold for a larger set of vertices which all have the same eigenvector centrality. By this we mean that if a network $G$, has a set $S$ of three or more vertices which all have the same eigenvector centrality, then there is no guarantee that $\mathcal{P}_S(G)$ contains a symmetry.  We also note here that the above theorem is only true in general for undirected graphs. For instance vertices 2, and 6 for the graph $G$ in Figure \ref{fig:directedex} have the same eigenvector centrality but are not symmetric in $\mathcal{P}_{\{2,3\}}(G)$.

\end{remark}
The reason we use isospectral reductions instead of the comparatively simple Perron complement is that by using isospectral reductions we have a smaller class of symmetries.  If there is a symmetry in the Perron complement graph $\mathcal{P}_S(G)$ this does not always correspond to a latent symmetry in $\mathcal{R}_S(G)$.  That is, the Perron complement cannot be used to find latent symmetries in general. Further, we do not gain any new information regarding network symmetries from the Perron complement graph for directed graphs where Theorem \ref{thm:centrality2} does not hold.

\section{Latent Symmetries and Network Growth Models}\label{sec:5}

Real-world networks are constantly evolving and typically growing (see \cite{gross2009adaptive} for a review of the evolving structure of networks). A number of network formation models have been proposed to describe the type of growth observed in these networks.  The purpose of this section is to demonstrate that, like standard symmetries, latent symmetries appear to be a hallmark of such networks.  To determine how likely it is to find a latent symmetry in a given network, we preform a number of numerical experiments. In these experiments we count how many latent symmetries occur in randomly generated graphs. We focus these numerical experiments on directed networks, where latent automorphisms are more likely to occur.

The most well-known class of network growth models are those related to the Barab\'asi-Albert model \cite{barabasi1999emergence} and its predecessor the Price model \cite{price1976general}. In these models elements are added one by one to a network and are preferentially attached to vertices with high degree, i.e. to vertices with a high number of neighbors or some variant of this process \cite{albert2000topology}, \cite{dorogovtsev2000scaling}, \cite{krapivsky2001degree}. These models are devised to create networks that exhibit some of the most widely-observed features found in real networks such as scale-free degree distributions.  We choose to generate networks using this theoretical model to understand whether latent symmetries are likely to appear in a real-network setting.

We choose a two-parameter growth model which is a variation of the one originally devised by Price (see \cite{newman2010networks}, section 14.1).  When generating a graph, we begin with a \emph{complete graph} on 3 vertices and add vertices and edges iteratively.  At each step we add one vertex and two edges.  Both of these edges connect to the new vertex on one end.  The other end of each new edge attaches to a vertex in the existing graph with a probability proportional to the number of edges already connected to it plus some intrinsic weight $\alpha$ given to each vertex, which is a parameter we vary in our experiment.  Vertices that already have many adjacent edges are more likely to attach to the new vertex at each step. The $\alpha$ parameter changes how strongly the new vertices are preferentially attached to vertices with larger degrees.  The idea is that as $\alpha$ gets larger the graph is generated in a less preferentially-attached way.

In these experiments, we generate 1000 graphs with 180 vertices at each value of $\alpha$. We then randomly determined which direction each edge points, where there is a 1/5 probability the edge will point in both directions, otherwise it points in only one direction.  We use this to method increases the connectivity of the resulting network. The reason for these probabilities is that this makes each edge half as likely to be directed in both directions compared to being directed in just one direction. That is, an edge is directed in one way with probability 2/5, and the other way with 2/5 probability, and in both ways with probability 1/5.  Once a graph is generated, we reduce the graph to its largest strongly connected component since networks are often analyzed at this level. The result is a collection of graphs with a mean number of about 137 vertices.  Finally, we count how many of the resulting graphs have at least one  standard symmetry and how many have at least one latent symmetry.

Figure \ref{fig:plots}  plots the percentage of graphs generated for each value of $\alpha$ which have a standard symmetry (left) and latent symmetry (right).  We first notice that both plots essentially decrease as $\alpha$ increases, demonstrating that as graphs are generated in a way less like preferential attachment, we find less symmetry at any scale. Though there are fewer total latent symmetries than real symmetries, they both follow the same general trend, suggesting the process of creating standard and latent symmetries are correlated.  This suggests that mechanisms which allow for a greater number of standard symmetries also allow for the formation of more latent symmetries. Exactly what this mechanism is and how it operates even in these experiments is an interesting and open question. One possibility is that having regions of low edge density among collections of vertices creates an environment where symmetries are more likely to occur randomly. Thus a method which generates a network using preferential attachment concentrates most of the connection around vertices with high degree (hubs), allowing other vertices to have a lower density of edges.

\begin{figure}
\includegraphics[scale=.4]{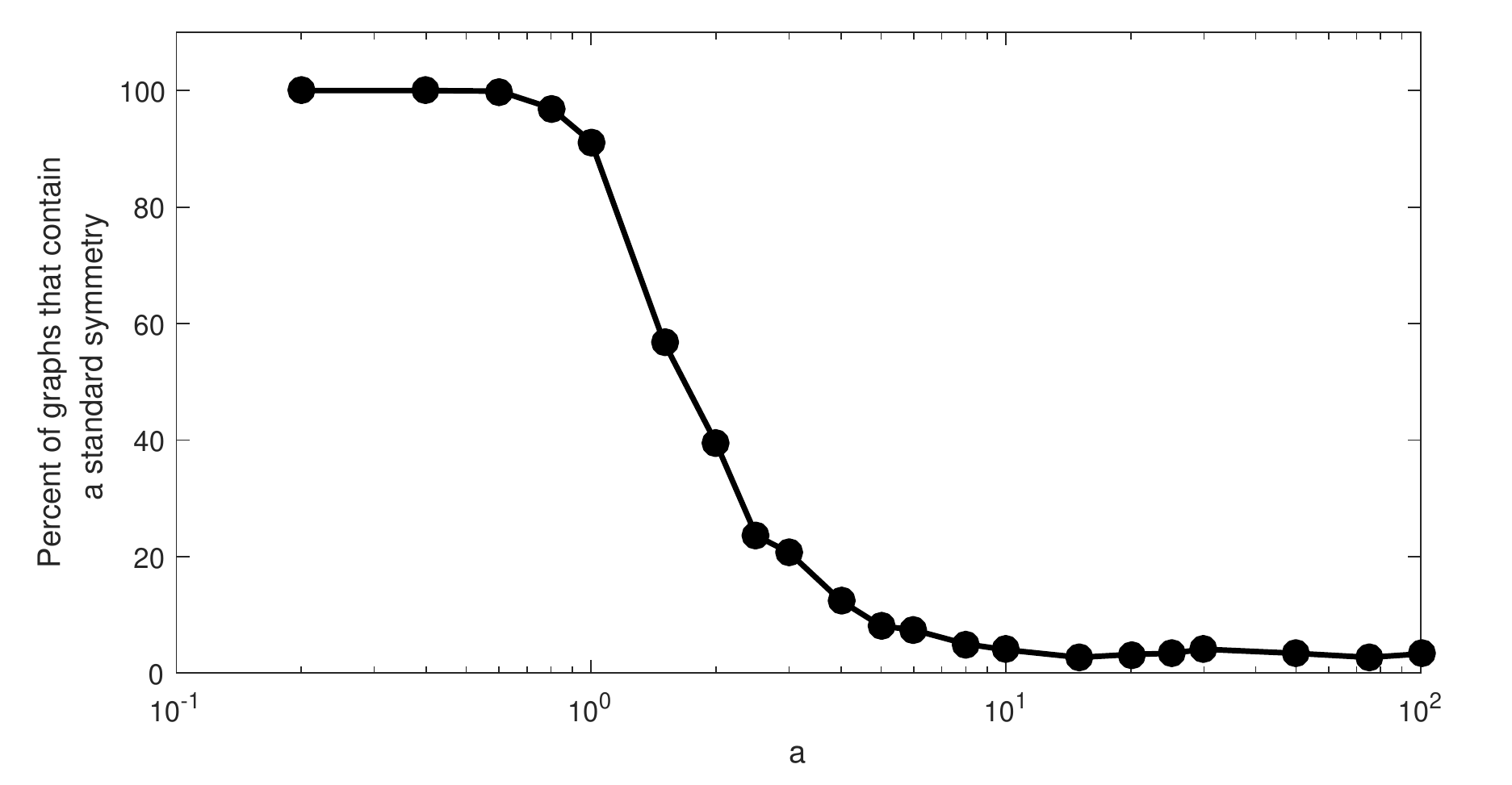}
\raisebox{.04\height}{  \includegraphics[scale=.34]{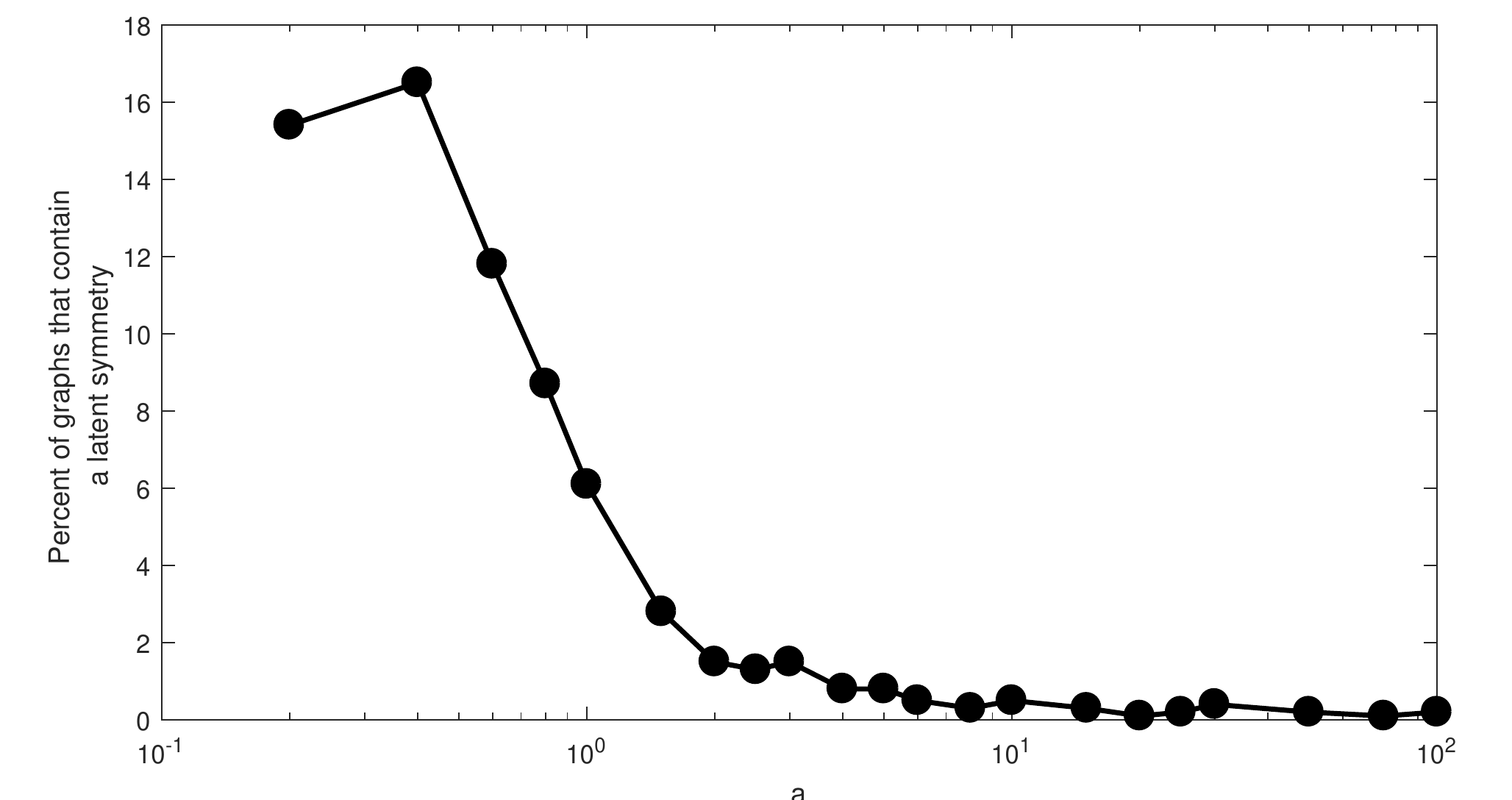}}\caption{The left figure plots the percentage of graphs which were generated using preferential attachment that contain a standard symmetry.  The horizontal axis (plotted logarithmically) gives different values of $a$, the parameter which controls how strongly each edge is attached preferentially. The right is the same figure in which the occurrences of latent symmetries are plotted.}\label{fig:plots}
\end{figure}

 To compare the results of our experiments, we do the same experiment by building graphs via the Erd{\H o}s-R\'enyi method \cite{erdos1959random}.  This is a method for generating random graphs where there is a fixed probability $\alpha$ that an edge will exist between any two vertices in the graph.  Importantly, it is known that this does not lead to graphs that resemble real-world networks \cite{watts1998collective}. Starting as before with 180 vertices, we generate Erd{\H o}s-R\'enyi graphs in which there is an edge between every pair of vertices in the graph with probability $p$. We then follow the same procedure of randomly directing edges and reducing the network to its largest strongly connected component as in the previous experiment where we considered preferential attachment. After generating 1000 graphs for a given value of $p$ we determine what percentage of these graphs have at least one structural symmetry and at least one latent symmetry as before.

The results of this experiments show there is a very narrow range of values for $p$, between $p=.02$ and $p=.03$, for which we find any symmetries and also generate graphs with a substantial largest connected component. For all values of $p$ we found that $5.9\%$ was the highest percentage of these graphs having a standard symmetry, and $0.3\%$ was the highest percentage of these graphs having a latent symmetry.  These numbers are much lower than anything we found using preferential attachment. These experiments were repeated for larger starting vertex numbers and similar results were obtained. This suggests that symmetries, and in particular latent symmetries, are not as likely to be found in nonnetwork-like graphs as they are in graphs that capture some important features of real-world networks.

\begin{remark} We note here that these experiments were only performed on directed networks.  We find that latent symmetries are extremely rare in undirected networks generated using either of the above models for a large number of vertices (more than fifteen). Interestingly, when searching for symmetries in small undirected graphs, we were actually able to find more latent symmetries using the Erd{\H o}s-R\'enyi model than with preferential attachment. Although most networks are much larger than fifteen vertices, it is not understood why the trend observed for large networks should reverse for small networks.

\end{remark}

\section{Conclusion}\label{sec:conclusion}
We have proposed a novel generalization of the notion of a symmetry for a network (graph), which reveals to an extent the latent structure of both real and theoretical networks. We define these latent symmetries to be structural symmetries in an isospectrally reduced version of the original network. Though we showed some simple examples of latent symmetries in a number of simple graphs, we demonstrated they naturally occur in real networks (see Section \ref{sec:3}).  In addition we also define a measure of latency of a symmetry which gives a sense of scale of the symmetry or how deep the symmetry is hidden in the network.  It is worth mentioning that this measure of latency can be difficult to compute as it requires searching through all possible vertex subsets of a graph. An interesting and open question is whether $\mathcal{M}(S)$ can be efficiently computed or approximated.

In the real-world networks we considered, and in the theoretical networks we numerically generated we were able to find latent symmetries that coexisted at various scales within the same network.  This seems to suggest that real-world networks are not only rich with symmetries \cite{Mac2008}, but have what we might term a \emph{hierarchical structure of symmetries} in which symmetries can be found at multiple scales within the network.

In section \ref{sec:4}, we showed that vertices in a network which are latently symmetric also have the same eigenvector centrality.  Thus, in this sense latently symmetric vertices share equal importance in the network. This suggests that latent symmetries are encoding some type of network structure which is invisible to a simple symmetry search. Theorem \ref{thm:centrality2} actually strengthens this result for the undirected case, showing that the Perron complement of the network over two vertices results in a symmetry if and only if those two vertices have the same eigenvector centrality in the original network.

One of the strongest cases for the importance of latent symmetries comes from the numerical study we perform in section \ref{sec:5}, which shows that structural symmetries and latent symmetries are correlated in graphs generated using preferential attachment.  This suggests that we should expect latent symmetries to naturally occur in real networks which form via some form of preferential attachment.

Having demonstrated the potential of latent symmetries as a concept for analyzing the structure of networks, many questions still remain.  For instance, do networks utilize latent symmetries the same way they utilize standard structural symmetries? More specifically, do latently symmetric nodes often have similar functions? Or could they have complimentary functions? If a set of vertices are latently symmetric and some subset of them fail, what happens to the network? i.e. does the network also fail in some way?

In the section \ref{sec:5}, we found latent symmetries are more likely to occur in a network built using a preferential attachment model.  It is natural to ask if there are other models of network growth which also lead to the formation of latent symmetries.  It was also noted that latent symmetries seem to be rare in large undirected networks. There is still work to be done to understand why this is and if there is a model which could result in latent symmetries in undirected graphs. In fact, the seeming dichotomy between directed and undirected graphs, which we find in our numerical and theoretical results, is not well understood and is a source of many open questions.

In summary, there is much more work to do regarding latent symmetries, but from our preliminary work it seems that understanding how and why latent symmetries form could provide clues to how network structures form and are utilized.

\section{Appendix}

In this section we prove Lemma \ref{v} which is needed in the proof of Theorem \ref{thm:EC}.  We also prove Theorem \ref{thm:centrality2} which connects latent symmetries in undirected graphs and the Perron complement graph.

\begin{lem}\label{v}
If $B$ is a nonzero square submatrix of $A$ where $A$ is nonnegative and irreducible, then $\rho(B) < \rho (A)$.
\end{lem}
\begin{proof}
Let $A$ be  an irreducible, nonnegative, $n\times n$ matrix with square submatrix $B$. Now we can write $A$ as $A=A_1+A_2$ where $A_1$ is the matrix $A$ with all entries corresponding to $B$ set to zero.  Now let $C_\epsilon=\epsilon A_1+ A_2$ and $D_\epsilon=(1-\epsilon) A_1$ for some $0<\epsilon <1$.  Thus $A=C_\epsilon+D_\epsilon$ where $C_\epsilon$ and $D_\epsilon$ are both nonnegative matrices, and $C_\epsilon$ is irreducible (this is inherited from $A$%, since we we cannot lose irreducibility by slightly adjusting some of the matrix values
) and $D_\epsilon$ is not zero (since B is not zero).  Now according to exercise 8.4.P14 in \cite{horn1990matrix},  $\rho(C_\epsilon+D_\epsilon)>\rho(C_\epsilon)$.  Next we note that $A_2=C_0$, where $C_0$ is $C_\epsilon$ with $\epsilon$ is set to $0$.  But $C_\epsilon> A_2$ for any $\epsilon$, thus we can use Theorem 8.1.18 in \cite{horn1990matrix} to guarantee that $\rho(C_\epsilon)\geq \rho(A_2)$.  Finally we notice that the because $A_2$ and $B$ only differ by rows and columns of zeros. Thus, except for zero eigenvalues, $A_2$ and $B$ share the same spectrum. Therefore, $\rho(A_2)=\rho(B)$. Putting these results together we see that $$ \rho(A)=\rho(C_\epsilon+D_\epsilon)>\rho(C_\epsilon)\geq\rho(A_2)=\rho(B).$$  Thus we have shown that $\rho(A)>\rho(B)$.
\end{proof}

\begin{thm}\textbf{(Theorem \ref{thm:centrality2})}
 Let $G=(V,E,\omega)$ be an undirected graph with adjacency matrix $A$.  Vertices $a$ and $b$ have the same eigenvector centrality if and only if they are symmetric in $\mathcal{P}_{\{a,b\}}(G)$.
\end{thm}
\begin{proof}
Suppose $a,b$ are symmetric in $\mathcal{P}_{\{a,b\}}(G)$, then $a,b$ have the same value in  $\mathbf{v}_{\{a,b\}}$, the leading eigenvector of $\mathcal{P}_{\{a,b\}}(A)$.  Thus by Theorem \ref{thm:irr} (ii), $a,b$ must also have the same value in $\mathbf{v}$, the leading eigenvector of $A$ since $\mathbf{v}_{\{a,b\}}$ is just a projection of $\mathbf{v}$.  Thus $a$ and $b$ must have the same eigenvector centrality in $G$.

Now suppose that  $a$ and $b$ have the same eigenvector centrality in $G$.  Thus for the unique largest eigenvalue, which in this case is the spectral radius, $\rho(G)=\lambda_0$ with corresponding eigenvector $\mathbf{v}$, we know that $\mathbf{v}_a=\mathbf{v}_b=v$, (where $\mathbf{v}_i$ is the entry in $\mathbf{v}$  corresponding to vertex $i$). Now using Theorem 1 in
\cite{duarte2015eigenvectors}, we know that $\lambda_0$ is also an eigenvalue of $\mathcal{R}_{\{ a,b\} }(G)$  with corresponding eigenvector ${{\bf{v}}_{\{ a,b\} }} = \left( {\begin{matrix}
   v  \\
   v  \\

 \end{matrix} } \right)$.  Now using the definition of the eigenpair we get
\begin{align}
  & {\mathcal{R}_{\{ a,b\} }}(A){|_{\lambda  = {\lambda _0}}}{{\bf{v}}_{\{ a,b\} }} = {\lambda _0}{{\bf{v}}_{\{ a,b\} }}  \cr
  & \left( {\begin{matrix}
   {{p_{11}}({\lambda _0})} & {{p_{12}}({\lambda _0})}  \\
   {{p_{21}}({\lambda _0})} & {{p_{22}}({\lambda _0})}  \\
 \end{matrix} } \right)\left( {\begin{matrix}
   v  \\
   v  \\
 \end{matrix} } \right) = \left( {\begin{matrix}
   {{\lambda _0}v}  \\
   {{\lambda _0}v}  \\
 \end{matrix} } \right)
 \end{align}
	
 However; because we started with a symmetric matrix (undirected graph) we know that the result of an isospectral matrix reduction is also symmetric (see \cite{thebook}. Thus ${p_{21}}(\lambda ) = {p_{12}}(\lambda )$ and
	  \begin{align}
  & \left( {\begin{matrix}
   {{p_{11}}({\lambda _0})} & {{p_{12}}({\lambda _0})}  \\
   {{p_{12}}({\lambda _0})} & {{p_{22}}({\lambda _0})}  \\
 \end{matrix} } \right)\left( {\begin{matrix}
   v  \\
   v  \\
 \end{matrix} } \right) = \left( {\begin{matrix}
   {{\lambda _0}v}  \\
   {{\lambda _0}v}  \\
 \end{matrix} } \right)  \cr
  \end{align}
Therefore we conclude that ${p_{11}}({\lambda _0}) = {p_{22}}({\lambda _0}) $ and the reduced matrix has the form
${\mathcal{R} _{\{ a,b\} }}(A){|_{\lambda  = \rho (G)}}$ has the form $$\left( {\begin{matrix}
   {{p_{11}}({\lambda _0})} & {{p_{12}}({\lambda _0})}  \\
   {{p_{12}}({\lambda _0})} & {{p_{11}}({\lambda _0})}  \\
 \end{matrix} } \right)$$which is symmetric between the two remaining vertices.
\end{proof}

\bibliographystyle{unsrt}
\bibliographystyle{plain}
\bibliography{bibfile}
\end{document}